\documentclass[aps,prx,twocolumn,footinbib,floatfix]{revtex4-2}
\usepackage{graphicx}
\usepackage{indentfirst}
\usepackage{braket}
\usepackage{float}
\usepackage{amsmath}
\usepackage{physics}
\usepackage{amssymb}
\usepackage{CJK}
\usepackage{ulem}
\usepackage{xcolor}
\usepackage{esint}
\usepackage{color}
\usepackage[T1]{fontenc}
\usepackage{subfigure}
\usepackage{amsfonts}
\usepackage{footmisc}
\usepackage{scrextend}
\usepackage{multirow}
\usepackage[outdir=./]{epstopdf}

\usepackage[english]{babel}
\usepackage{url}
\usepackage{comment}
\usepackage{array}
\usepackage{subfiles}
\usepackage{tikz}
\usetikzlibrary{shapes,arrows}
\usepackage{mathrsfs}
\usepackage[mathscr]{eucal}

\usepackage{url}

\usepackage[hyperfootnotes=false]{hyperref}
\usepackage{cleveref}
\usepackage{natbib}
\definecolor{darkblue}{rgb}{0,0,0.5}
\hypersetup{
colorlinks=true,
linkcolor=black,
filecolor=blue,
citecolor=darkblue,  
urlcolor=black,
}

\newtheorem{theorem}{Theorem}

\newtheorem{corollary}{Corollary}

\newtheorem{lemma}{Lemma}

\newenvironment{proof}[1][Proof]{\noindent\textbf{#1.} }{\ \rule{0.5em}{0.5em}}

\def\be{\begin{equation}}
\def\ee{\end{equation}}
\def\ba{\begin{eqnarray}}
\def\ea{\end{eqnarray}}
\def\bal{\begin{equation}\begin{aligned}}
\def\eal{\end{aligned}\end{equation}}

\def\bp{\begin{pmatrix}}
\def\ep{\end{pmatrix}}

\usepackage{bm}

\newcommand{\calD}{{\cal D}}

\newcommand{\1}{^{(1)}}

\newcommand{\state}[1]{\ketbra{#1}{#1}}

\usepackage{pict2e,picture,graphicx}
\makeatletter
\DeclareRobustCommand{\Arrow}[1][]{%
\check@mathfonts
\if\relax\detokenize{#1}\relax
\settowidth{\dimen@}{$\m@th\rightarrow$}%
\else
\setlength{\dimen@}{#1}%
\fi
\sbox\z@{\usefont{U}{lasy}{m}{n}\symbol{41}}%
\begin{picture}(\dimen@,\ht\z@)
\roundcap
\put(\dimexpr\dimen@-.7\wd\z@,0){\usebox\z@}
\put(0,\fontdimen22\textfont2){\line(1,0){\dimen@}}
\end{picture}%
}
\makeatother

\usepackage{soul}

\usepackage{xr}
\begin{document}

\title{Quantum-enhanced learning with a controllable bosonic variational sensor network}
\author{Pengcheng Liao${}^{1}$}

\author{Bingzhi Zhang${}^{1,2}$}

\author{Quntao Zhuang${}^{1,2}$}
\email{qzhuang@usc.edu}

\affiliation{
${}^{1}$Ming Hsieh Department of Electrical and Computer Engineering, University of Southern California, Los
Angeles, California 90089, USA
\\
${}^{2}$Department of Physics and Astronomy, University of Southern California, Los
Angeles, California 90089, USA
}

\begin{abstract}

The emergence of quantum sensor networks has presented opportunities for enhancing complex sensing tasks, while simultaneously introducing significant challenges in designing and analyzing quantum sensing protocols due to the intricate nature of entanglement and physical processes. Supervised learning assisted by an entangled sensor network (SLAEN) [Phys. Rev. X 9, 041023 (2019)] represents a promising paradigm for automating sensor-network design through variational quantum machine learning. However, the original SLAEN, constrained by the Gaussian nature of quantum circuits, is limited to learning linearly separable data. Leveraging the universal quantum control available in cavity-QED experiments, we propose a generalized SLAEN capable of handling nonlinear data classification tasks. We establish a theoretical framework for physical-layer data classification to underpin our approach. Through training quantum probes and measurements, we uncover a threshold phenomenon in classification error across various tasks---when the energy of probes exceeds a certain threshold, the error drastically diminishes to zero, providing a significant improvement over the Gaussian SLAEN. Despite the non-Gaussian nature of the problem, we offer analytical insights into determining the threshold and residual error in the presence of noise. Our findings carry implications for radio-frequency photonic sensors and microwave dark matter haloscopes.
\end{abstract}
\maketitle


\section{Introduction}

Quantum sensors leverage resources such as squeezing and entanglement to enhance sensitivity, dynamical range, and bandwidth. In canonical sensing scenarios like single-sensor optical interferometric phase sensing, the Gaussian nature of the problem~\cite{weedbrook2012gaussian} not only facilitates the intuitive design of squeezing injection to enhance sensitivity~\cite{caves1981quantum,ganapathy2023broadband}, but also enable well-calibrated theoretical sensitivity predictions to guide experimental efforts~\cite{demkowicz2013fundamental,escher2011general}.

As quantum sensor capabilities advance, quantum advantage can now extend beyond single sensors to distributed quantum sensor networks~\cite{zhuang2018distributed,ge2018distributed,zhang2021distributed}, addressing complex and dynamic sensing scenarios in optical phase sensing~\cite{guo2020distributed,liu2021distributed}, RF photonic sensing~\cite{xia2020demonstration}, optomechanical force sensing~\cite{xia2023entanglement}, and networks of atomic clocks~\cite{malia2022distributed}. However, the rise of complex sensor networks involving multipartite entanglement presents significant challenges in designing and analyzing optimal quantum sensing protocols. 

To address the challenge, Ref.~\cite{zhuang2019physical} proposed applying variational quantum circuits (VQC)~\cite{cerezo2021variational}---an important tool in near-term quantum computing---to adaptively train an entangled sensor network for classifying physical-layer data modeled as bosonic quantum channels. The proposed supervised learning assisted by an entangled sensor network (SLAEN) was later experimentally verified~\cite{xia2021quantum}, demonstrating quantum advantage in physical-layer data classification. However, the original SLAEN relies on non-universal Gaussian quantum resources~\cite{weedbrook2012gaussian} and is limited to reducing classification errors in linearly separable data, thereby limiting its practical relevance given the prevalence of nonlinear data in common machine learning tasks.

In this study, we develop a controllable bosonic variational sensor network for nonlinear physical-layer data classification, empowered by universal quantum control available in various experimental systems~\cite{eickbusch2022fast,diringer2023conditional,fluhmann2019encoding,hauer2023nonlinear,gerashchenko2024towards}. Across various classification tasks, we observe a threshold phenomenon in classification error, wherein the error abruptly decreases to zero at a specific probe energy threshold. This phenomenon, supported by our theoretical foundation on displacement data classification, offers a substantial advantage over previous Gaussian SLAEN approaches~\cite{zhuang2019physical,xia2021quantum}. Despite the non-Gaussian nature of the problem, we provide analytical results to determine the threshold and demonstrate robustness to noise. We apply the controllable bosonic variational sensor network to two examples, modeling microwave dark matter haloscopes~\cite{backes2021quantum,shi2023ultimate,brady2022entangled} and radio-frequency (RF) photonic sensors~\cite{xia2020demonstration}, and demonstrate the quantum advantage enabled by the threshold phenomenon.


We begin our paper with the description of the protocol in Section~\ref{sec:scheme} and then proceed to the general theory framework in Section~\ref{sec:theory}. Then we begin considering examples of linear separable data in Section~\ref{sec:linear_case} and nonlinear data in Section~\ref{sec:nonlinear}. Finally, we conclude with discussions in Section~\ref{sec:conclusion}.

\section{Controllable bosonic variational sensor networks}
\label{sec:scheme}

Unlike various proposals for quantum machine learning algorithms aimed at solving classical problems, the SLAEN approach~\cite{zhuang2019physical} addresses quantum data classification within physical processes, a domain that has recently garnered significant attention in the quantum machine learning community~\cite{huang2022quantum,banchi2021generalization,caro2023out,gebhart2023learning}. Depending on the sensor types employed, the physical-layer data in the form of bosonic quadrature displacements can originate from diverse scenarios. For instance, in radar detection, RF signals can be transduced into the optical domain as quadrature displacements, which are measurable by optical sensor networks~\cite{xia2020demonstration}; Optomechanical sensors~\cite{li2021cavity,liu2021progress,xia2023entanglement}, based on optical-mechanical coupling, can detect force~\cite{gavartin2012hybrid}, acceleration~\cite{krause2012high}, and magnetic fields~\cite{forstner2012cavity} by measuring quadrature displacements. Moreover, fundamental physics sensing experiments can also be represented as quadrature displacement data classification tasks. The Laser Interferometer Gravitational-Wave Observatory (LIGO)~\cite{ganapathy2023broadband}, for instance, transduces gravitational wave signals into optical phase shifts, which are then converted into quadrature displacements. Similarly, in microwave dark matter haloscopes\cite{Sikivie:1983ip,backes2021quantum,brady2022entangled}, potential dark matter candidates such as axions induce random excess noise in cavities, which can be modeled as random quadrature displacements. These examples highlight optical quadrature displacement (at different wavelengths) as a universal model for such physical-layer data, leading us to focus on optical quadrature displacement classification in this study.

\begin{figure}[t]
    \centering
    \includegraphics[width=\columnwidth]{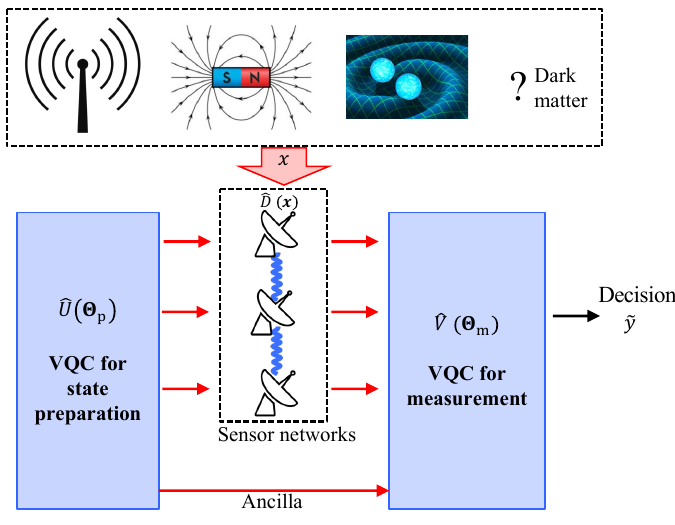}
    \caption{Schematic of the SLAEN empowered by universal control.}
    \label{fig:concept}
\end{figure}

As proposed in Ref.~\cite{zhuang2019physical}, the architecture of SLAEN is depicted in Fig.~\ref{fig:concept}, where two VQCs are involved to empower a sensor network for physical-layer data classification. The first VQC realizes a parameterized unitary $\hat{U}({\bm \theta}_p)$ to prepare the entangled probe for the sensor network; while the second VQC implements the quantum measurement for decision-making via another parameterized unitary $\hat{V}({\bm \theta}_m)$ and standard measurements, such as homodyne, heterodyne or single-qubit Pauli measurements. In between the two VQCs, physical-layer data induces a unitary $\hat{D}(\bm x)$ on the probe via the quantum sensor network, transducing information from other domains onto the quantum states. In this work, we focus on the displacement operator $\hat{D}(\bm x)\equiv \otimes_{m=1}^M \hat{D}(q_m+i p_m)$ with $\bm x=(q_1,p_1,\cdots, q_M,p_M)$ being the displacements on position and momentum quadratures, which model the various applications introduced above. Here $\hat{D}(\beta) \equiv \exp(\beta \hat{a}^\dagger-\beta^* \hat{a})$ is the single-mode quadrature displacement operator on the mode described by annihilation operator $\hat{a}$ and we can define the position and momentum quadrature operator as $\hat q=({\hat a+\hat a^\dagger})/{\sqrt{2}}$ and $\hat p=({\hat a-\hat a^\dagger})/{\sqrt{2}i}$.

\begin{figure}[t]
    \centering
    \includegraphics[width=\columnwidth]{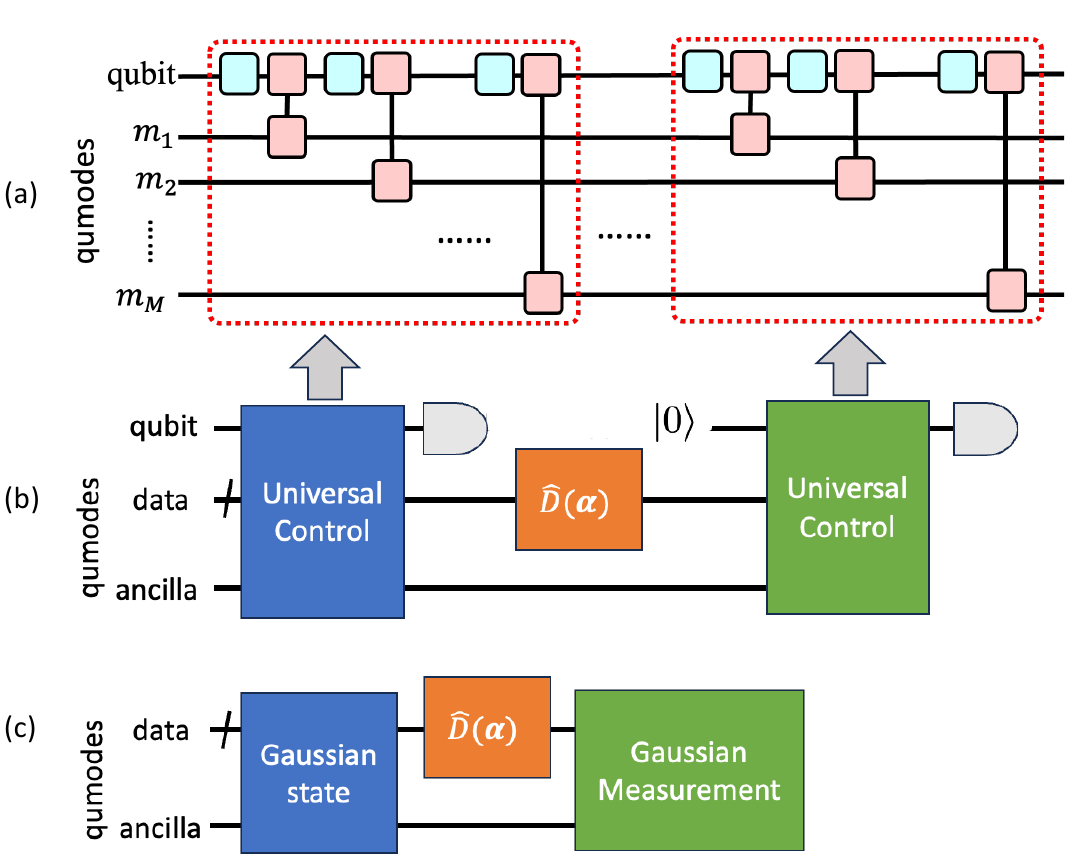}
    \caption{(a) Schematic of the universal control circuit over a qubit coupled to $M$ qumodes $m_1,\cdots, m_M$. The cyan boxes represent single-qubit rotations and the pink connected boxes represent qubit-qumode coupling such as the echoed conditional displacement (ECD) gates. Additional control qubits can also be introduced if necessary. In (b) and (c), we show schemes to learning displacement data $\hat{D}(\alpha)$ via the general universal control VQC and Gaussian scheme separately.}
    \label{fig: ECD-circuit}
\end{figure}

To go beyond the previous generation of SLAEN, we allow entanglement ancillae that are stored parallel to the sensing process of the sensor network. More importantly, while both unitaries $\hat{U}({\bm \theta}_p)$ and $\hat{V}({\bm \theta}_m)$ in the previous SLAEN are from linear optical quantum circuits in the class of Gaussian unitaries~\cite{weedbrook2012gaussian}, we propose to adopt universal quantum control to tackle nonlinear data classification. Such universal control is well supported by experimental efforts in various quantum information processing tasks, such as quantum computing and bosonic quantum error correction~\cite{brady2023advances}. In the microwave domain, thanks to the strong nonlinearity, universal quantum control with echoed conditional displacement (ECD) gates~\cite{eickbusch2022fast} and conditional-NOT displacement gates~\cite{diringer2023conditional} is readily available in cavity quantum electrodynamics (QED) experiments, suitable for empowering microwave dark matter haloscopes~\cite{Sikivie:1983ip,backes2021quantum,brady2022entangled}. In the optical frequencies, universal control based on optical cavity-QED is also proposed~\cite{hastrup2022protocol} in theory, paving the way to enhance optical quantum sensors. The universal control also extends to mechanical resonators~\cite{gerashchenko2024towards}, which further broadens the scope of applicability of our results to mechanical signal processing and testing of fundamental physics of quantum gravity~\cite{penrose1996gravity}.

As different approaches of universal control are equivalent by definition, we will consider the ECD gate without loss of generality, where a single qubit is introduced to control the displacement of a quantum mode (qumode)
\be 
\hat{U}_{\rm ECD}(\beta) = \hat{D}(\beta)\otimes \ket{1}\bra{0} + \hat{D}(-\beta)\otimes\ket{0}\bra{1}.
\ee 
ECD gates enable universal control over a system consisting of multiple qubits and qumodes~\cite{zhang2023energy}, when
combined with single qubit rotations $\hat{U}_{\rm R}(\xi, \phi) = \exp[-i\xi/2(\cos\phi\hat{\sigma}^x + \sin\phi\hat{\sigma}^y)]$, where $\hat{\sigma}^x$ and $\hat{\sigma}^y$ are Pauli-X and Y operators. For $M$ modes, each of the VQCs $\hat{U}({\bm \theta}_p)$ and $\hat{V}({\bm \theta}_m)$ has the structure depicted in Fig.~\ref{fig: ECD-circuit}, where $L$ layers of ECD gates are applied. To benchmark the advantage of entanglement ancillae in supervised learning, we consider both entanglement-assisted (EA) VQC and unassisted VQC---denoted as EA-VQC and VQC correspondingly.

The set of training parameters ${\bm \theta}_p$ and ${\bm \theta}_m$ each contain the displacements $\{\beta_\ell\}_{\ell=1}^L$ and the qubit rotation angles $\{\xi_\ell, \phi_\ell\}_{\ell=1}^L$. 
To enable the testing between $H$ hypotheses, we can consider having $\lceil\log_2 H \rceil$ qubits to be measured in the computational bases on the output side of the VQC to provide the classification $\tilde{y}$. For simplicity, we denote the states $\{\ket{\tilde{y}}\}$ with integers $\tilde{y}\in[0,H-1]$ as the computation bases of $\lceil\log_2 H \rceil$ qubits directly, without specifying the tensor structure.

To train the variational parameters $\bm \theta_p, \bm \theta_m$ towards correct classification, we consider a set of training data $\calD\equiv \{ \hat{D}(\bm x_k), y_k\}_{k=1}^N$, where $y_k\in \{0,1,\cdots, H-1\}$ is the class label and $N$ is the total number of data. For displacement $\hat{D}(\bm x_k)$, the output quantum state of the overall system before measurement is
\be 
\ket{\Phi(\bm x_k;{\bm \theta}_m,{\bm \theta}_p)}=\hat{V}({\bm \theta}_m)\hat{D}(\bm x_k)\hat{U}({\bm \theta}_p)\ket{vac}\otimes \ket{0},
\ee 
where $\ket{vac}$ means all bosonic modes are initially vacuum and $\ket{0}$ means all qubits involved are initially in the zero state.
The probability of measuring the result $\tilde{y}_k$ on the output state is therefore
\be 
P(\tilde{y}_k|\bm x_k)= 
\expval{\left(\state{\tilde{y}_k}\otimes \hat{I}\right)}{\Phi(\bm x_k, {\bm \theta}_m,{\bm \theta}_p)},
\ee 
where identity $\hat{I}$ is over all remaining systems including bosonic modes and any ancillary qubits.

Given a set of parameters $\bm \theta_p, \bm \theta_m$ for the VQC and the data $\calD$, the error probability of the hypothesis testing can be expressed as: 
\be 
P_{\rm E}(\hat{U}({\bm \theta}_p), \hat{V}({\bm \theta}_m);\calD,N_S)=1- \mathbb{E}_{(\bm x_k,y_k)\sim \calD} P(\tilde{y}_k=y_k|\bm x_k).
\label{pe_N_data}
\ee 
To consider practically relevant strategies, we also introduce energy regularization on the input bosonic modes to the sensors---the total mean occupation number $\sum_{k=1}^M  \expval{\hat{m}_k^\dagger\hat{m}_k}\le N_S$, where the expectation value is over the state $\hat{U}({\bm \theta}_p)\ket{vac}\otimes \ket{0}$. Since the physical process of interest is qudrature displacements, we can choose equality in the regularization without loss of generality, as each lower energy state has the equivalent performance to a displaced version of the state with energy exactly $N_S$.

\section{Data transform and threshold phenomenon}
\label{sec:theory}

Before optimizing VQCs $\hat{U}({\bm \theta}_p)$ and $\hat{V}({\bm \theta}_m)$, we also hope to obtain analytical insights into how the performance depends on the data structure. For our theory analyses, we can describe the data $\calD$ via a joint probability distribution $P_{XY} (\bm x;y)$ of displacement $\hat{D}(\bm x)$ and the label $y$.
We are interested in solving the minimum Bayesian error probability 
\begin{align}
P_{\rm E}^\star (\calD,N_S)=\min_{\hat{U}, \hat{V}} P_{\rm E}(\hat{U},\hat{V}; \calD, N_S),
\label{PE_opt}
\end{align}
where $P_{\rm E}$ is the error probability in Eq.~\eqref{pe_N_data}.

To obtain analytical insights into data structure, we make use of the property of displacements under Gaussian unitary pre-processing and post-processing. For a zero-mean Gaussian unitary $U_{\bm{S}}$ described by the $2M\times 2M$ symplectic matrix $\bm{S}$, we have~\cite{zhuang2019Scrambling}
\be
\hat{U}^\dagger_{\bm{S}} \hat{D}(\bm{\zeta}) \hat{U}_{\bm{S}} =  \hat{D} (\bm{S}^{-1} \bm{\zeta}).
\ee    
Such an affine coordinate transformation can be adopted to provide reduction between different data distributions, $P_{XY}^\prime (\bm x,y) = |S|P_{XY}(S \bm x,y)$; however, such a transform may not be free, as the energy of the quantum state prior to the data displacement also changes in a state-dependent fashion. In Appendix~\ref{app:Gaussian_transform}, we prove a general result on data transform in Theorem~\ref{the:app}, detailing the above intuition in a rigorous manner. Fortunately, in some cases, we are able to make use of such a transform to obtain useful insights, when ancillae are introduced (see Appendix~\ref{app:proof_theorem1} for a proof).

\begin{lemma}
\label{theorem:pe}
(stronger signal reduces error in EA-VQC)
Denote $P_{\rm E}^\star(\calD,N_S)$ as the optimal error probability of the EA-VQC in the discrimination between multiple ensembles of displacement operators, where the data $\calD$ of $M$-mode displacements is described by a joint distribution $P_{XY}$, under the overall input energy constraint $\sum_{k=1}^M\expval{\hat{m}_k^\dagger \hat{m}_k}\le N_S$. We have the relationship
\be 
P_{\rm E}^\star(\calD^\prime, N_S^\prime)\le P_{\rm E}^\star(\calD, N_S),
\label{pe_inequality}
\ee 
where the new data $\calD^\prime$ has the distribution $P^\prime_{XY}(\bm x;y)=c^{2M} P_{XY}(c \bm x;y)$ linearly transformed and $N_S^\prime=c^2N_S+M\max[c^2-1,0]$.
\end{lemma}

The consequence of Lemma~\ref{theorem:pe} is intuitive: the error probability monotonically decreases as the `signal' strength increases under the same probe energy constraint if entanglement ancillae are allowed. When $c<1$, the new problem has increased displacements by a scaling factor $1/c>1$. At the same time, the energy constraint in the new problem $N_S^\prime =c^2 N_S$ is lowered, therefore $P_{\rm E}^\star(\calD^\prime,N_S)\le P_{\rm E}^\star(\calD^\prime,N_S^\prime)\le P_{\rm E}^\star(\calD,N_S)$, which tells us that a better signal provides a lower error probability. As intuitive as it may seem, such reduction of error probability is not always guaranteed when no entanglement ancillae are allowed, as we detail in the later part of the paper (also see Ref.~\cite{gorecki2022quantum}).

More importantly, the intuitive decrease of error probability implies that if the EA-VQC can achieve zero error probability at a particular signal strength, then the error probability remains zero when the signal further increases upon training---supporting a threshold phenomenon of error probability. To emphasize this consequence, we summarize it in the following corollary.
\begin{corollary}
\label{theorem:EA_threshold}
(EA-VQCs support a threshold phenomenon) If an EA-VQC achieves zero error probability at a signal strength, then the error probability will remain zero when signal strength further increases.
\end{corollary}

\section{Linear separable case: threshold phenomenon}
\label{sec:linear_case}

We begin the study with a simple toy problem of binary displacement classification, where the two candidate displacements $\pm \epsilon$ are real without loss of generality. Due to the simplicity of the problem, we expect entanglement ancillae to be unnecessary in the optimal strategy and therefore focus on the VQC strategy (without EA) in comparison to Gaussian strategies. The classification error probability for different $\epsilon$ is shown in Fig.~\ref{fig:binary_n=1}(a). In the small $\epsilon\ll 1$ region, we expect the error probability of all protocols to be close to $1/2$, including the Gaussian scheme (orange) with squeezed state input and homodyne measurement (see Fig.~\ref{fig: ECD-circuit} bottom) and the VQC results (blue). 
With the increase in $\epsilon$, the Gaussian strategy error probability $P_{\rm E}=\frac{1}{2}{\rm Erfc}(\sqrt{2}\epsilon/e^{-r})\sim \exp\left(-2e^{2r}\epsilon^2\right)$ decays with the increasing signal amplitude continuously (orange), where the squeezing strength $r$ is related to energy regularization as $N_S=\sinh^2r$~\cite{zhuang2019physical}. In contrast, the VQC scheme demonstrates a threshold phenomenon where $P_{\rm E}=0$ is achieved up to numerical precision ($\lesssim 10^{-15}$) when $\epsilon$ is above a critical value $\epsilon_{\rm th}$. In this simple toy problem, the increasing signal $\epsilon$ directly provides a monotonically decreasing error probability for VQC, even without entanglement assistance (cf. Theorem~\ref{theorem:pe} and corollary~\ref{theorem:EA_threshold}). In later examples, we will see cases where entanglement assistance is required to guarantee the monotonic decrease. The threshold phenomenon is unique to the controllable VQC with non-Gaussian resources, while absent in strategies based on Gaussian resources such as squeezing. Indeed, even when provided with photon counting (green) or the general measurement that reaches the Helstrom limit~\cite{Helstrom_1967,Helstrom_1976} (red dashed), the squeezed state has a continuously decreasing finite error probability as the signal strength $\epsilon$ increases.

To gain insight into the VQC scheme, we plot the Wigner functions and photon number distributions of the optimal probe states $\hat{U}({\bm \theta}_p)\ket{vac}\otimes \ket{0}$ before the sensor-induced displacement in Fig.~\ref{fig:binary_n=1} (c) and (d) for various different values of $\epsilon$. The Wigner functions show an asymmetric peak resembling that of a squeezed vacuum; however, the tail of the squeezed vacuum is highly non-Gaussian which eliminates the Gaussian error tail. In particular, we find that for a fixed value of $\epsilon>\epsilon_{\rm th}$, the optimal probe state that leads to zero error is not unique, while below the threshold, the converged optimal states are identical up to numerical convergence error for each value of $\epsilon$.

\begin{figure}
    \centering
    \includegraphics[width=\linewidth]{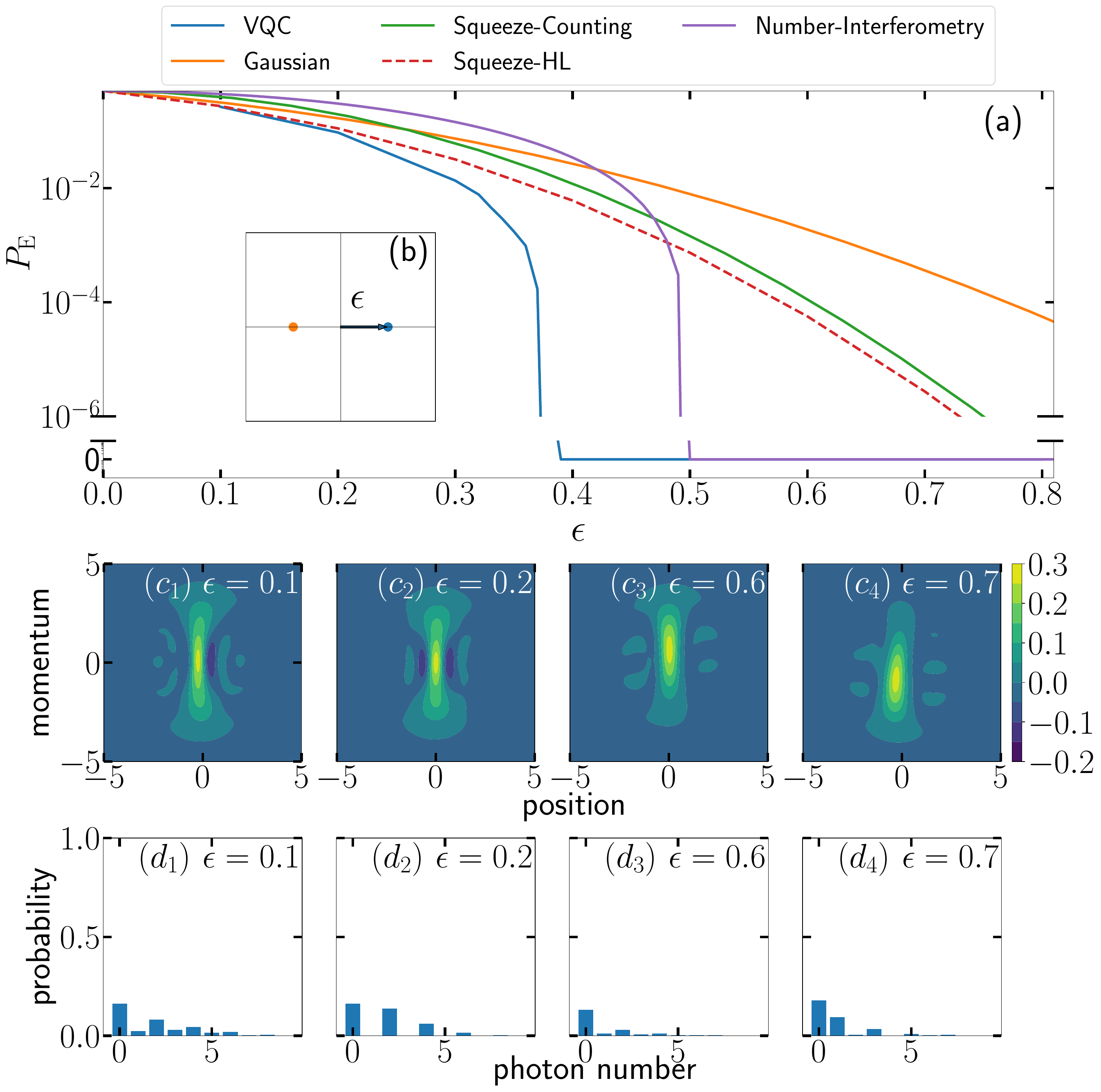}
    \caption{Binary displacement discrimination with probe state energy $N_S=1$. (a) Error probability versus displacement amplitude $\epsilon$. Here `0' in $P_{\rm E}$ refers to error probability $\lesssim 10^{-15}$ reaches the numerical precision. Inset (b) shows the two hypothesis of displacement $\pm \epsilon$. (c1)-(c4) show the Wigner function of optimal probe states from VQC for different $\epsilon$. (d1)-(d4) show their corresponding photon number distribution.}
    \label{fig:binary_n=1}
\end{figure}

Following the similarity to squeezed vacuum at small $\epsilon$, we also provide a first-order theory to explain the critical $\epsilon_{\rm th}$ for a given mean photon number $N_S$. For any input probe state $\ket{\psi}$, the two displacements map the probe state $\ket{\psi}$ to $\ket{\psi_\pm}=\hat{D}(\pm \epsilon)\ket{\psi}$. In the small $\epsilon$ limit, we have the overlap
\begin{align}
    |\expval{\psi_+|\psi_-}|^2&= \left\lvert\expval{\psi|\hat{D}^\dagger (-\epsilon) \hat{D}(\epsilon)|\psi} \right\rvert^2 
    \\&= 1 - 4\epsilon^2 \expval{\hat{p}}^2 + \mathcal{O}(\epsilon^3).
\end{align}
Zero-error discrimination requires the states to be orthogonal, $|\expval{\psi_+|\psi_-}|^2=0$, which leads to $\epsilon^2 \expval{\hat{p}}^2 = 1/4$ to the leading order. Recall that the mean occupation number is $N_S = \left(\expval{\hat{q}^2} + \expval{\hat{p}^2} - 1 \right)/2$ for a single mode. Due to uncertainty principle $\expval{\hat{q}^2} \cdot \expval{\hat{p}^2}\ge 1/4$ for zero-mean quantum state, we have the first-order estimation on the critical signal $\epsilon_{\rm th}$ as
\begin{align}
    \epsilon_{\rm th} = \frac{1}{2 \sqrt{\max\expval{\hat{p}^2}}} = \left[\frac{1 + 2N_S - 2\sqrt{N_S(N_S + 1)}}{2}\right]^{1/2}.
    \label{eth_asymp}
\end{align}
In Fig.~\ref{fig:binary_error_threshold_at_different_energy}(b), we see agreement in the scaling between Eq.~\eqref{eth_asymp} and the numerical results, where the thresholds (yellow dots) are identified in VQC training for various $N_S$ values as shown in Fig.~\ref{fig:binary_error_threshold_at_different_energy}(a). 

\begin{figure}
    \centering
    \includegraphics[width=1.0\linewidth]{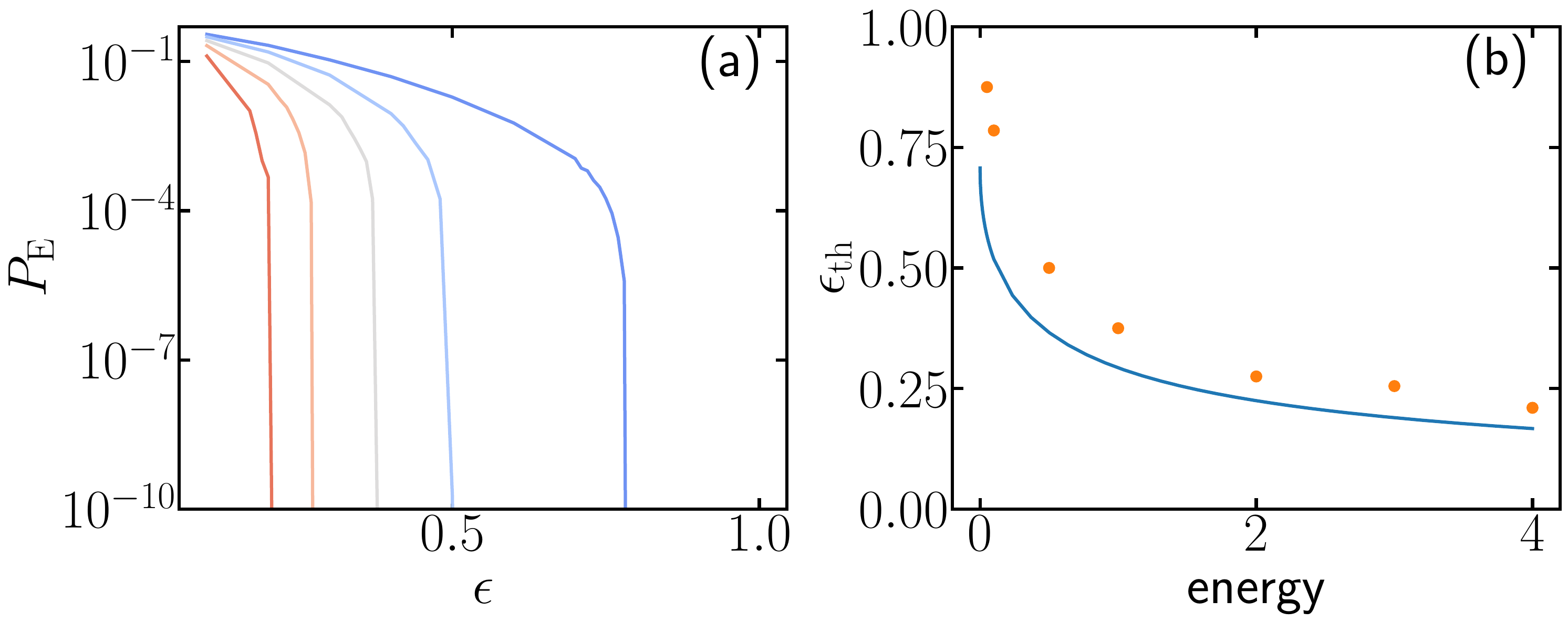}
    \caption{(a) Error probability $P_{\rm E}$ of optimal state from VQC versus displacement amplitudes $\epsilon$ under various energy constraint $N_S = 0.1,0.5,1,2,4$ (from blue to red). (b) The corresponding threshold $\epsilon_{\rm th}$ versus energy $N_S$. Orange dots are numerical VQC results and the blue solid line is asymptotic results in Eq.~(\ref{eth_asymp}).}
    \label{fig:binary_error_threshold_at_different_energy}
\end{figure}

The threshold phenomenon represents a large error-probability gap between the VQC strategy and Gaussian quantum strategy or other classical strategies. 
Now we address the robustness of the threshold phenomenon to noise to certify the practical advantage. The standard noise model for the continuous-variable modes is random distribution of displacements, as typically considered in bosonic quantum error correction~\cite{gkp2001,brady2023GKPrvw}. We consider a more general noise model where a random unitary $\hat{U}_{\bm \zeta}=\exp(-i{\bm \zeta}\cdot \hat{\bm g})$ is applied on the probe, with a list of generators $\hat{\bm g}=\{\hat{g}_j\}$. Without loss of generality, we can assume the distribution of random error amplitude $P(\bm \zeta)$ is zero-mean. With the random noise, the two hypotheses can be described by two quantum channels with the input-output relation
\begin{align}
\Phi_\pm(\hat{\rho}) = \int {\rm d}\bm \zeta\, P(\bm \zeta) \hat{U}_{\bm \zeta} \left[\hat{D}(\pm \epsilon) \hat{\rho} \hat{D}^\dagger(\pm \epsilon)\right] \hat{U}^\dagger_{\bm \zeta}.
\label{Phi_noisy}
\end{align}
In the case of real displacement noise, a single generator $\hat{g}_1=\hat{q}$ is sufficient. In the case of complex displacement noise, the generators are $\{\hat{g}_1,\hat{g}_2\}=\{\hat{q},\hat{p}\}$. 
We are interested in whether the error probability in the above-threshold region will remain close to zero when the noise is small. Quantitatively, we obtain the following theorem (see Appendix~\ref{app:proof_noisy_pe} for a proof).
\begin{theorem}
\label{noisy_theorem}
For $\epsilon>\epsilon_{\rm th}$ past the threshold of zero error, in the weak noise limit of equal-prior binary displacement discrimination between quantum channels $\Phi_\pm$ in Eq.~\eqref{Phi_noisy}, the minimum error probability
\be 
P_{\rm E}^\star(\calD, N_S)\le \sum_{i, j} \operatorname{Cov}(\zeta_i, \zeta_j) \operatorname{Cov}(\hat{g}_i, \hat{g}_j  ),
\label{pe_varg_general}
\ee 
where $\operatorname{Cov}(\zeta_i,\zeta_j)=\int {\rm d}\bm\zeta\, P(\bm\zeta) \zeta_i\zeta_j$ is the covariance of the noise distribution and $\operatorname{Cov}(\hat{g}_i, \hat{g}_j  )$ is the average covariance $\expval{\frac{1}{2}\{\hat{g}_i ,\hat{g}_j\}}-\expval{\hat{g}_i}\expval{\hat{g}_j}$ over the optimal state in the ideal noiseless case.
In the case of single generator $\hat{g}$, Ineq.~\eqref{pe_varg_general} reduces to
\be 
P_{\rm E}^\star(\calD, N_S)\le \sigma^2_{\rm noise}{\rm Var}(\hat{g}) ,
\label{pe_varg}
\ee 
where $\sigma^2_{\rm noise}=\int {\rm d}\zeta\, P(\zeta) \zeta^2$ is the variance of the zero-mean noise distribution and ${\rm Var}(\hat{g})$ is the average variance of $\hat{g}$ for the optimal state in the ideal noiseless case.
\end{theorem}

\begin{figure}[t]
    \centering
    \includegraphics[width=\linewidth]{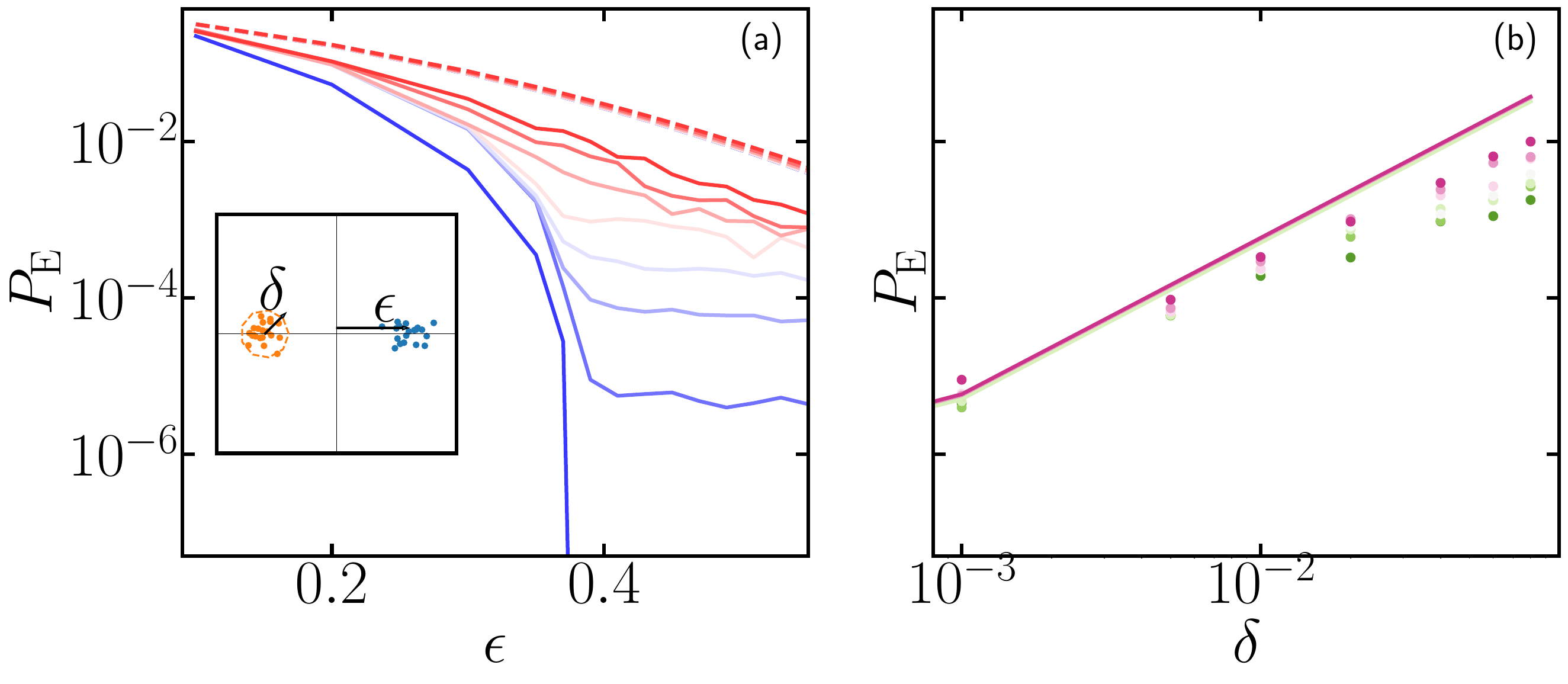}
    \caption{Error probability vs displacement amplitude $\epsilon$ and distribution standard deviation $\delta$ in the one-dimensional linearly separable case, where the energy constraint is $N_S=1$. The training data is sampled from 2D Gaussian distribution with mean $(\pm \epsilon,0)$ (class 1/0), variance $\delta^2$ and zero covariance. The $\delta$ values in (a) are $\{0, 0.001, 0.005, 0.01, 0.02, 0.04, 0.06, 0.08\}$ (from blue to red) and the $\epsilon$ values in (b) are $\{0.39, 0.41, 0.43, 0.45, 0.47, 0.49, 0.51\}$ (from purple to green). In (a), the solid and dashed lines correspond to the performance of the universal VQC and the Gaussian protocol.
    In (b), the solid lines are the upper bounds in Eq.~\eqref{pe_varg_general}, which are highly overlapped in the figure.}
    \label{fig:1d_linear}
\end{figure}

As mentioned earlier, a common noise in the class of Eq.~\eqref{Phi_noisy} is the additive Gaussian noise with generators $\{\hat{q},\hat{p}\}$. Additive noise comes from excess thermal noise and amplification of loss in bosonic systems~\cite{brady2023GKPrvw}. In addition, additive noise can come from uncertainty of the prior information, as exact knowledge of the two hypotheses is often not possible. In other cases, one can consider phase noise, where the generator is the number operator $\hat{n}$. From Theorem~\ref{noisy_theorem}, we also provide the advantage compared to Gaussian strategies in the presence of noise.
For other strategies (e.g. Gaussian strategies) that do not possess the threshold phenomenon, error probability can be decreased by performing $R$ repetitions of experiments if the signal is static, i.e., $P_{\rm E} \propto \exp(-R c)$ for some constant $c$. In this case, then the number of repetitions in such a strategy to match up with the noisy performance in Eq.~\eqref{pe_varg} is $\propto \ln(\sigma^2_{\rm noise})$, which can be large when the noise is low. Moreover, for transient signals, such repetitions will not be possible and the error probability gap indicated in Fig.~\ref{fig:binary_n=1} remains huge in that case. 

Now we consider an example to compare the upper bound in Theorem~\ref{noisy_theorem} and the achievable numerical results for noisy binary displacement discrimination. For a given value of $\epsilon$, we consider different noise displacement distributions---complex Gaussian distribution with variance $\delta^2$ on each quadrature. Therefore, the distribution of the actual amplitude of displacements are the noisy version of $\pm \epsilon $ as indicated in Fig.~\ref{fig:1d_linear}(a) inset. We train the VQC under the same energy constraint in the presence of the random noise and observe the increase of the error probability from zero to finite values, as expected (see Fig.~\ref{fig:1d_linear}). The advantage over the Gaussian scheme survives in all noise levels evaluated, as high as noise over signal strength $\delta/\epsilon\sim 16\%$. To compare with the error probability upper bound in Eq.~\eqref{pe_varg_general}, we plot the error probability versus the noise standard deviation $\delta$. To evaluate Eq.~\eqref{pe_varg_general}, we take the optimal state in the ideal case from the VQC training in Fig.~\ref{fig:binary_n=1} and calculate the variances of quadratures, ${\rm Var}(\hat{q})$ and ${\rm Var}(\hat{p})$ correspondingly ($\expval{\hat{q}\hat{p}}$ is not needed since the noise distribution has zero cross-correlation). The obtained upper bounds agree with the numerical results in terms of the scaling $P_{\rm E}\propto \delta^2$. Moreover, for $\epsilon\sim 0.39$ close to the critical threshold, the upper bound (red solid) agrees well with the actual performance (red dots) and therefore is tight. Also we see the upper bound is generally tighter for smaller $\delta^2$ as expected. At $\epsilon=0.39$ and $\delta=10^{-3}$ (left-most purple dot in Fig.~\ref{fig:1d_linear}(b)), we see a slight violation of the upper bound, because $P_E^\star$ is not achieved in the VQC training due to numerical convergence issue.

In the multi-mode case, when the real displacements are clustered in one direction to be linear separable but spread in all other orthogonal directions, one can transform the displacement data through dimensionality reduction methods into the corresponding one-dimensional classification problems, and thus we expect the threshold phenomenon also applies there.

\section{Nonlinear displacement classification}
\label{sec:nonlinear}

In this section, we extend our study to the classification of displacement data that are not linearly separable. Specifically, we consider two examples that represent real-world applications to demonstrate the power of the proposed VQC framework. 
The first application concerns the fundamental question of the nature of dark matter, which is indicated by astrophysical and cosmological observations and simulations~\cite{aghanim2020planck,massey2010dark,sofue2001rotation}. To test ultralight dark matter candidates such as axion, a major approach is via microwave-cavity detection of excess noise that may be induced by axion-photon coupling in the presence of a magnetic field~\cite{Sikivie:1983ip,girvin2016axdm,malnou2019,dixit2021,backes2021}. The excess noise can be modeled as quadrature displacements of the cavity modes. While a Gaussian distributed displacement model is convenient for theoretical analyses~\cite{shi2023ultimate}, at a fixed frequency, axion signals can be better modeled as a displacement with a fixed amplitude $\epsilon$ and an entirely random phase~\cite{gorecki2022quantum}. In Section~\ref{sec:complex_DM}, we will apply the VQC to discriminate between such a uniform circle of complex displacement versus zero displacement---a hypothesis testing model for dark matter detection.

\begin{figure}[t]
    \centering
    \includegraphics[width=\linewidth]{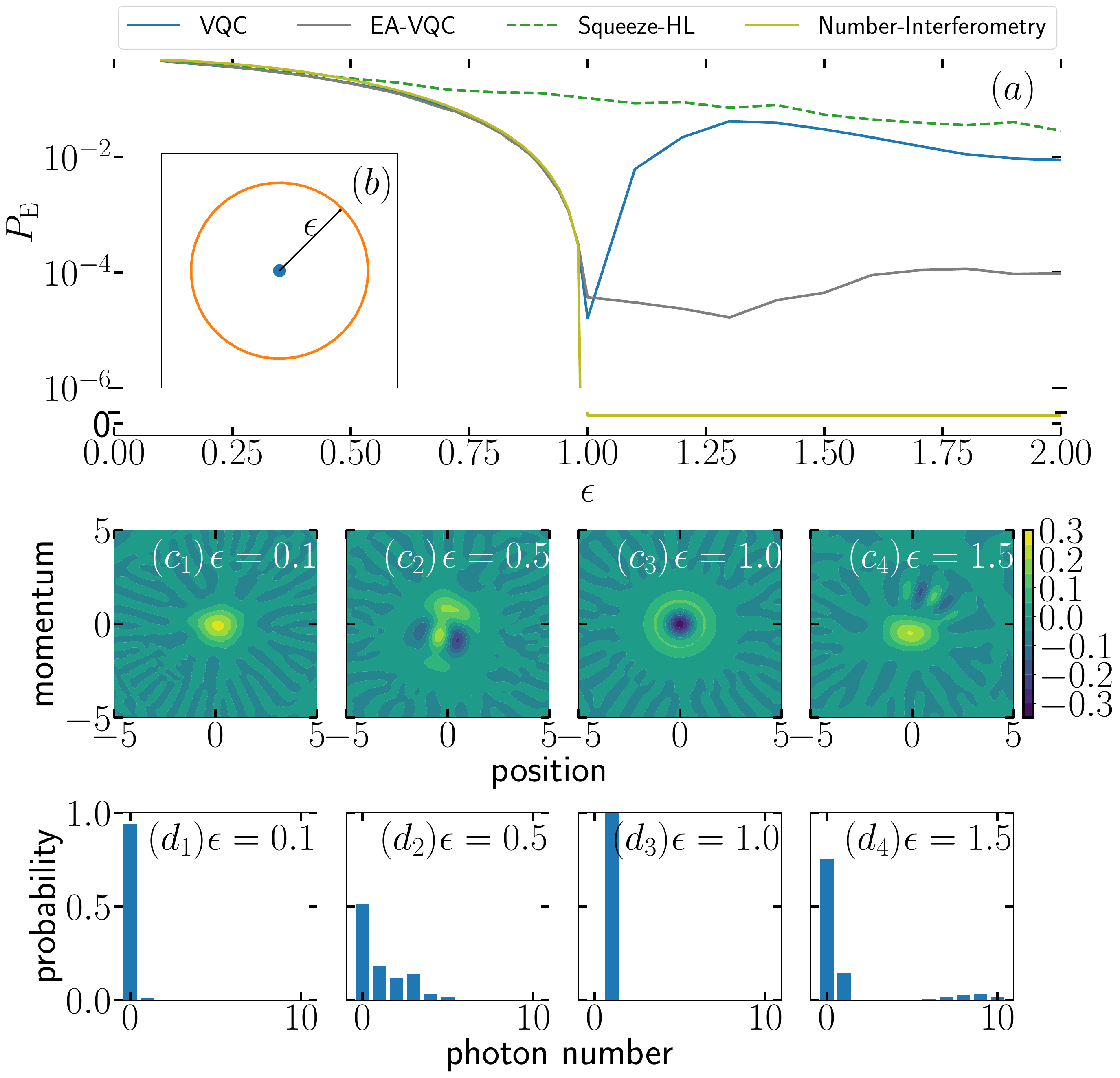}
    \caption{(a) Classification error versus displacement amplitude $\epsilon$ for identity versus circular data under energy constraint $N_S=1$. Performance of VQC on training and testing data agree very well and only the training results are presented here. Here `0' in $P_{\rm E}$ refers to error probability $\lesssim10^{-15}$ reaches the numerical precision. The displacement distributions are shown in $(b)$, in which one class is identity and another class is uniformly distributed circle. Some of the optimal probe states at different $\epsilon$ obtained from VQC are shown in $(c)$ and $(d)$. When $\epsilon$ is small, the optimal state is close to vacuum $(c_1)$ and $(d_1)$; when $\epsilon=1.0$, the optimal state is number state $\ket{1}$ ( $(c_3)$ and $(d_3)$). 
    }
    \label{fig:1d_complex_concentric_circle}
\end{figure}
 
The second application involves transducer-based photonic sensing, including RF-photonic sensors~\cite{xia2020demonstration,ghelfi2014fully} and optomechanical sensors~\cite{xia2023entanglement,liu2021progress}. In these sensors, the signal from various sources modulates the optical coherent state pump, and therefore generates displacements on the side-band. While the displacements are in general complex, here we consider real displacements to model a case where phases are aligned between sensors, as considered in previous works~\cite{xia2020demonstration,xia2023entanglement}. There, Gaussian quantum resources such as squeezing and multipartite Gaussian entanglement are shown to enable quantum advantages over classical sensors. Here, we explore the VQC performance in Section~\ref{sec:RF_app}.

\subsection{One-mode complex displacements: with application in dark matter detection}
\label{sec:complex_DM}

Now we begin with the first example of nonlinear displacement hypothesis testing modeling dark matter detection, where displacement $\alpha$ can be either zero  or uniformly distributed in a circle $|\alpha| = \epsilon$, as shown in Fig.~\ref{fig:1d_complex_concentric_circle} (b). 
In the limit of $\epsilon \to 0$, we expect the optimal state to be the single-mode squeezed vacuum, as hypothesis testing reduces to parameter estimation when $\epsilon\to 0$~\cite{gorecki2022quantum}. As $\epsilon$ increases, a greater separation becomes evident between the VQC schemes (blue and gray solid) and the squeezing-enabled Gaussian scheme (green dashed). At a critical value of $\epsilon_{\rm th}$, we also identify the threshold phenomenon, similar to the linear separable case. Here, the error probability (blue and gray solid) sharply decreases to as low as $10^{-5}$, limited by the convergence precision of VQC training. The optimal state at the threshold $\epsilon_{\rm th}$ turns out to be a number state, as indicated by Fig.~\ref{fig:1d_complex_concentric_circle}(c3)(d3). In the VQC without EA (blue solid), however, we find the error probability increases with increasing $\epsilon$ past the threshold ($\epsilon\ge \epsilon_{\rm th}$). 
Despite the increasing error probability for the VQC without EA (blue solid), it is instructive to consider the optimal states past the threshold. As indicated by the Wigner functions and photon number distributions in Fig.~\ref{fig:1d_complex_concentric_circle}(c4)(d4), a superposition between two photon number states emerges as the optimal state. Indeed, we can evaluate the Helstrom error probability limit of the `0N' state
$
    \ket{\rm 0N} \propto\left( w\ket{0} + \ket{N}\right)
$
which agrees well with the VQC results (blue solid) for large $\epsilon$ when $N$ and $w$ are optimized under the energy constraint. As indicated by Corollary~\ref{theorem:EA_threshold}, EA-VQC is generally required to support the threshold phenomenon, as observed by the gray solid line, where the error probability of the EA-VQC stays around $10^{-5}$. The confirmed separation manifests the importance of entanglement-assistance in general displacement sensing problems.

Inspired by the optimality of photon number state and ON state obtained from VQC training, we design a number-interferometry protocol to achieve close-to-optimal performance. As indicated in Lemma~\ref{theorem:pe}, to support the threshold phenomenon, we can introduce an ancilla and a beamsplitter to effectively reduce the case with $\epsilon>\epsilon_{\rm th}$ to the threshold case. As shown in Fig.~\ref{fig:Number_state_interferometry_Scheme}, a photon number state is interfered with a vacuum to probe the displacement and then an inverse of the initial beamsplitter reverts one output port to the initial number state when the displacement is zero. Finally, photon counting is performed on the corresponding output port to determine the hypothesis.
\begin{figure}
    \centering
    \includegraphics[width=1.0\linewidth]{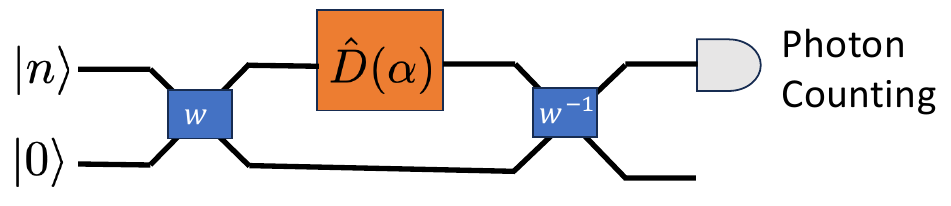}
    \caption{Schematic of number state interferometry.}
    \label{fig:Number_state_interferometry_Scheme}
\end{figure}
The decision of the hypothesis can be made with a threshold decision strategy, with POVMs $\Pi_0 = \ketbra{n}{n}$ and $\Pi_1 = I- \Pi_0$. 
From the prior of the two hypotheses, the error probability of the number-interferometry protocol can be obtained as
\be
P_{\rm E} = \frac{1}{2}  \int_{|\alpha|=\epsilon}d^2 \alpha  |\expval{n|D(\alpha)|n} |^2,
\ee
where $|\alpha|=\epsilon$ indicates the uniform distribution with fixed amplitude $\epsilon$ representing dark matter present case.
The overlap has analytical solution,
$
\expval{\hat D(\alpha)}{n} = e^{-|\alpha|^2/2} L_n(|\alpha|^2),
$
where $L_n(x)$ is the Laguerre polynomials. To achieve zero error probability, we require $\expval{\hat D(\epsilon)}{n}=0$. For energy constraint $n = N_S=1$, $L_1(x) = 1-x$ has root 1, indicating the zero probability at $\epsilon=1$. This explains the threshold value in Fig.~\ref{fig:1d_complex_concentric_circle}. 

As the zeros of the Laguerre polynomial $L_n(x)$ are always positive, we only need to focus on the smallest zero $x_1^{(n)}$, which is found to be bounded by $1/n \le x_1^{(n)}\le 2/(n+1)$~\cite{skovgaard1951greatest}. 
Therefore, for the given probe state $\ket{n}$ combined with the necessary EA strategy, as shown in Fig.~\ref{fig:Number_state_interferometry_Scheme}, one can achieve zero error probability with $\epsilon \sim \Omega(1/\sqrt{n})$; In other words, for given prior uniform circular distribution of $\alpha$ with radius $\epsilon$, one needs the probe number state satisfying $n = \lceil 2/\epsilon^2-1 \rceil$ to enable the zero-error threshold phenomenon. We evaluate the performance of the number interferometry protocol as the green line in Fig.~\ref{fig:1d_complex_concentric_circle}(a), which indeed has threshold at $\epsilon_{\rm th}=1$ since $N_S=1$. Moreover, it also agrees well with the EA-VQC training result, though VQC results is slightly better than number interferometry before the critical threshold.

The goal of our VQC is to assist the design of state preparation and measurement to enable the best performance. In this example, we are able to distill from the VQC results and design the near-optimal number state interferometry scheme. However, in general hypothesis testing problems, coming up with an intuitive design is challenging and VQC provides a versatile tool for such problems. For example, when applying the number-state interferometry scheme to the simple binary displacement testing problem in Fig.~\ref{fig:binary_n=1}, the performance is much worse than VQC results.

\subsection{Real multi-dimensional displacement: with application in RF sensing}
\label{sec:RF_app}

\begin{figure}[t]
    \centering
    \includegraphics[width=\linewidth]{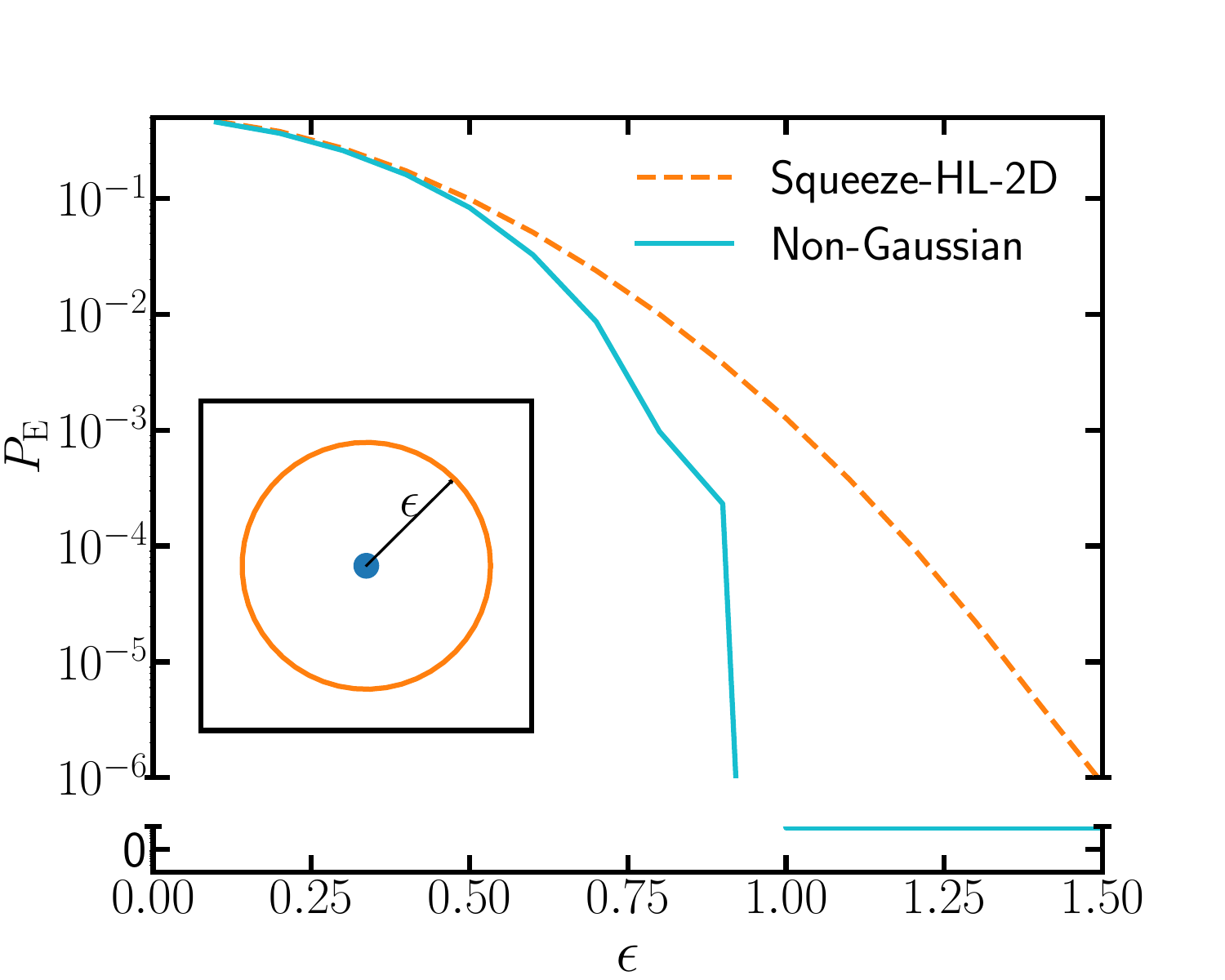}
    \caption{Error probability vs displacement amplitude $\epsilon$ in the two-dimensional linearly non-separable case, where the energy constraint is $N_S=1.0$. Here `0' in $P_{\rm E}$ refers to error probability $\lesssim10^{-15}$ reaches the numerical precision. The inset shows the data shape. The orange dashed curve is the Helstrom limit with squeezed vacuum input and the solid cyan curve is the non-Gaussian error probability enabled by VQC and number-interferometry.}
    \label{fig:2d_nonlinear}
\end{figure}

Now we study the classification error for nonlinear separable multi-dimensional real displacement, which is often encountered in RF sensing as described by Ref.~\cite{xia2020demonstration}. Limited by the numerical simulation power, we consider a two-dimensional case, while larger system size can be considered in the future with on-device training. For time-dependent RF signals, at each fixed time $t$, we consider data as
\begin{align}
    \left\{ \begin{array}{ll}
            x_1(t)  & = A_1 \cos(\omega t+\phi),\\
            x_2(t) &= A_2 \cos(\omega t +\phi+ \Delta\phi), 
    \end{array} \right.
\end{align}
where $\omega$ is the angular frequency of the RF signal, $\phi$ is a phase factor, $\Delta \phi$ is the relative phase, and $A_1$ and $A_2$ are the amplitudes determined by pump power and RF signal strength. 
When both the phase and relative phase are fixed, the time trajectory of the above two-dimensional signal typically manifests as an ellipse, which simplifies to a circle when $A_1$ equals $A_2$. When the phase $\phi$ of the RF signal is random, e.g. due to phase noise induced by path fluctuation or turbulence of RF signal propagation, the data of the RF signal becomes an ensemble of displacements distributed along the path. In this study, we focus on distinguishing between two hypotheses of the phase-random RF signal: the origin ($A_1 = A_2 = 0$) and a larger circle encompassing it ($A_1 = A_2 = \epsilon$), akin to prior investigations in one-dimensional complex systems. 

Illustrated in Fig. \ref{fig:2d_nonlinear}, a threshold phenomenon in the error probability is again evident for the non-Gaussian scheme (combining VQC and number-interferometry), where the error probability diminishes to zero when $\epsilon\ge \epsilon_{\rm th}$.
Thanks to the data reduction technique introduced prior to Lemma~\ref{theorem:pe}, in Appendix~\ref{app:data_reduction} we show that the two-dimensional real data classification problem can be reduced to a one-dimensional complex classification task via the SUM gate. Consequently, the results derived from the one-dimensional complex scenario are readily applicable to the two-dimensional real case, with appropriate adjustments to the energy constraints. For instance, the number-interferometry technique remains applicable, which leads to the threshold phenomenon in Fig.~\ref{fig:2d_nonlinear}.
The Non-Gaussian curve showcases the minimal error achieved by both the two-mode VQC and the number-interferometry technique. These approaches leverage non-Gaussian properties and outperform Gaussian schemes, such as the one employing squeezed states as input, as depicted in the figure.

\section{Conclusions and discussions}
\label{sec:conclusion}

Our work generalizes previous works of supervised learning assisted by an entangled sensor network (SLAEN)~\cite{zhuang2019physical,xia2021quantum} to the general controllable case and reveals an error-probability threshold phenomenon. We establish the fundamental theory to support the threshold phenomenon and asymptotic analyses to predict the threshold and noise robustness. 
Our findings provide a general method for resolving bosonic sensor design problem and hold implications for various quantum sensing scenarios, including optical phase sensing~\cite{guo2020distributed}, radiofrequency photonic sensors~\cite{xia2020demonstration}, optomechanical sensors~\cite{xia2023entanglement}, and microwave dark matter haloscopes~\cite{backes2021quantum,shi2023ultimate,brady2022entangled}. As future directions, our framework can also be applied to optimize measurement design for decoding coherent states in classical optical communication. While previous works~\cite{becerra2013experimental,cui2022quantum} consider adaptive photon counting schemes, our universal control scheme provides a much more powerful approach.
Finally, we discuss connections to some related works. Besides Gaussian SLAEN~\cite{zhuang2019physical,xia2021quantum}, variational quantum sensing has been a focus of various recent works~\cite{kaubruegger2019variational,meyer2021variational,kaubruegger2023optimal} with qubit-based systems and parameter estimation tasks. Refs.~\cite{sinanan2023single,oh2024entanglement} also considers bosonic systems; however, there the threshold phenomenon is not identified.

\begin{acknowledgements}
QZ thanks for discussions with Roni Harnik, Hsin-Yuan Huang, Liang Jiang and Konrad Lehnert. The project is supported by Office of Naval Research Grant No. N00014-23-1-2296, National Science Foundation OMA-2326746, National Science Foundation CAREER Award CCF-2240641 and National Science Foundation Engineering Research Center for Quantum Networks Grant No. 1941583. 
\end{acknowledgements}

\appendix

\section{Proof of Theorem~\ref{noisy_theorem}}
\label{app:proof_noisy_pe}

In the ideal binary displacement discrimination case, suppose the initial state is $\ket{\psi}$, then the two states need to be discriminated after the unknown displacement are
\be
\ket{\psi_-} = \hat{D}(-\epsilon)\ket{\psi};  \quad \ket{\psi_+} = \hat{D}(\epsilon)\ket{\psi}.
\ee
Moreover, suppose the POVM for hypothesis testing is $\hat\Pi_-$ and $\hat\Pi_+$ for the $-\epsilon$ and $+\epsilon$ displacements, then the error probability is
\be
P_{\rm E} = 1- \operatorname{Tr}(\ket{\psi_-}\bra{\psi_-}\hat\Pi_-)/2-
\operatorname{Tr}(\ket{\psi_+}\bra{\psi_+}\hat\Pi_+)/2, 
\ee
where we assume equal prior probability.
When $\epsilon\ge \epsilon_{\text{th}}$, the error probability is zero, indicating
\be
\operatorname{Tr}(\ket{\psi_-}\bra{\psi_-}\Pi_-) = \operatorname{Tr}(\ket{\psi_+}\bra{\psi_+}\Pi_+) = 1,
\ee
which leads to $\ket{\psi_-}\bra{\psi_-}=\Pi_-$ and $\ket{\psi_+}\bra{\psi_+}=\Pi_+$ . 

Consider the noisy quantum channels in Eq.~\eqref{Phi_noisy}, we have the achievable error probability
\begin{align}
P_{\rm E}&=1-\frac{1}{2}\sum_{k=\pm }\operatorname{Tr}\left\{\Pi_k\int {\rm d}\bm \zeta\, P(\bm \zeta) \hat{U}_{\bm \zeta} \state{\psi_k} \hat{U}^\dagger_{\bm \zeta}\right\}
\\
&=1-\frac{1}{2}\sum_{k=\pm } \int {\rm d}\bm \zeta\, P(\bm \zeta)\operatorname{Tr}\left\{\Pi_k \hat{U}_{\bm \zeta} \state{\psi_k} \hat{U}^\dagger_{\bm \zeta}\right\}.
\label{pe_full}
\end{align}
We take the $\zeta\ll1$ limit and expand  $U_{\bm{\zeta}} = e^{-i \bm{\hat{g}} \cdot\bm{\zeta} }$, where  $\bm{\hat{g}} \cdot\bm{\zeta} =  \sum_i \hat{g}_i \zeta_i$,
to its second-order approximation
\begin{align}
U_{\bm \zeta} \approx \hat{I}-i \hat{\bm{g}} \cdot\bm{\zeta} -( \hat{\bm g} \cdot\bm{\zeta})^2/2.   
\end{align}
As a result, the trace term in Eq.~\eqref{pe_full} can be expanded as
\begin{align}
&\operatorname{Tr}\left(\Pi_k \hat{U} \state{\psi_k} U^\dagger\right)\\
\approx&1+  \left(  \expval{ \psi_k \left| \sum_i \hat{g}_i \zeta_i \right| \psi_k }^2  - \expval{ \psi_k \left|  \left(\sum_i \hat{g}_i \zeta_i\right)^2 \right| \psi_k } \right)\\
 =& 1+\sum_i \zeta_i^2  \left( \braket{\psi_k |\hat{g}_i}{\psi_k}^2  - \braket{\psi_k |\hat{g}_i^2}{\psi_k}\right)\nonumber\\
 +& \sum_{i<j} 2 \zeta_i\zeta_j \left( \braket{\psi_k |\hat{g}_i}{\psi_k} \braket{\psi_k|\hat{g}_j}{\psi_k} - \frac{1}{2} \braket{\psi_k|\{\hat{g}_i,\hat{g}_j\}}{\psi_k}  \right)\\
 =& 1- \sum_i \zeta_i^2  \operatorname{Var}_k (\hat{g}_i) - \sum_{i<j} 2\zeta_i\zeta_j \operatorname{Cov}_k(\hat{g}_i,\hat{g}_j),
\end{align}
where ${\rm Var}_k(\hat{g}_i) = \expval{\psi_k|\hat{g}_i^2|\psi_k} - \expval{\psi_k|\hat{g}_i|\psi_k}^2$ is the variance of $
\ket{\psi_k}$ associated with $\hat{g}_i$, and $\operatorname{Cov}_k(\hat{g}_i, \hat{g}_j  ) = \expval{\frac{1}{2}\{\hat{g}_i ,\hat{g}_j\}}-\expval{\hat{g}_i}\expval{\hat{g}_j}$ is the covariance between $\hat{g}_i, \hat{g}_j$ for the state $\ket{\psi_k}$.

Assume the distribution $P(\bm \zeta)$ of noise $U_{\bm{\zeta}}$ has zero mean $\overline{\bm{\zeta}} = \bm{0}$, we can input the above into Eq.~\eqref{pe_full} to obtain the error probability 
\begin{align}
P_{\rm E}&
=\frac{1}{2} \sum_{k=\pm} \int {\rm d}\bm{\zeta}  \, P(\bm{\zeta} ) \left\{ \sum_i \zeta_i^2  \operatorname{Var}_k (\hat{g}_i) \right. \nonumber\\
 &\left. \quad\quad\quad\quad\quad\quad + \sum_{i<j} 2\zeta_i\zeta_j \operatorname{Cov}_k(\hat{g}_i,\hat{g}_j) \right\}\\
&=\frac{1}{2}\sum_{k=\pm} \left( \sum_i \sigma^2_{\zeta_i } \operatorname{Var}_k (\hat{g}_i) + \sum_{i<j} 2 \overline{\zeta_i\zeta_j} \operatorname{Cov}_k(\hat{g}_i,\hat{g}_j)  \right)
\\
&=\sum_{i}  \sigma^2_{\zeta_i}  \operatorname{Var}(\hat{g}_i) +2\sum_{i< j} \operatorname{Cov}(\zeta_i,\zeta_j) \operatorname{Cov}(\hat{g}_i, \hat{g}_j  ),
\\
 &= \sum_{i, j} \operatorname{Cov}(\zeta_i,\zeta_j) \operatorname{Cov}(\hat{g}_i, \hat{g}_j  ), 
\end{align}
leads to Eq.~\eqref{pe_varg_general}. Here $\sigma^2_{\zeta_i}$ is the variance of $\zeta_i$, and ${\rm Var}(\hat{g}_i) = 1/2\sum_{k=\pm } \operatorname{Var}_k (\hat{g}_i)$ and ${\rm Cov}(\hat{g}_i, \hat{g}_j) = 1/2\sum_{k=\pm } \operatorname{Cov}_k (\hat{g}_i, \hat{g}_j)$ are the average of variance and covariance for $\hat{g}_i$ over hypothesis state $\ket{\psi_\pm}$.

\section{Gaussian Transformation of Data}
\label{app:Gaussian_transform}
\label{app:proof_theorem1}
\label{app:data_reduction}

\begin{theorem}
\label{the:app}
 Denote $P^\star_{\rm E}(P_{XY},N_S)$ as the optimal error probability in the discrimination between multiple ensembles of $M_1$-mode displacement operators as specified in Eq.~\eqref{PE_opt}. 
Consider any $2(M_1+M_2)$-by-$2(M_1+M_2)$ symplectic matrix $\bm{S} = \begin{pmatrix}
 \bm{S}_{11} & \bm S_{12} \\\bm{S}_{21} & \bm{S}_{22}   
\end{pmatrix}$, where $\bm{S}_{11}\in \mathbb{R}^{2M_1\times 2M_1}$, $\bm{S}_{12}\in \mathbb{R}^{2M_1\times 2M_2}$, $\bm{S}_{21}\in \mathbb{R}^{2M_2\times 2M_1}$ and $\bm{S}_{22}\in \mathbb{R}^{2M_2\times 2M_2}$. Consider a new set of data distribution $P'_{XY}$ and energy constraint $N_S'$,
\begin{align}
P_{XY}^\prime (\bm{x},y) = \frac{1}{|\bm S_{22}|} P_{XY} \left[\left(\bm S_{11} - \bm S_{12}\bm S_{22}^{-1} \bm S_{21}\right)\bm x,y\right].
\label{PE_transform_full}
\end{align}
Then we have
\be
P^\star_{\rm E}(P_{XY},N_S)\le P^\star_{\rm E}(P_{XY}^\prime,N_S^\prime),
\label{general_ineq}
\ee   
where the energy
\begin{subequations}
\label{eq:energy_constraint}
\begin{align}
&N_S = \frac{1}{2}\Tr(\bm S_{11} \bm V' \bm S_{11}^T+\frac{1}{2}\bm S_{12}\bm S_{12}^T-M_1),
\\
&N_S^\prime =\frac{1}{2}\Tr(\bm V'-M_1),
\end{align}
\end{subequations}
are function of the covariance matrix $\bm V'$ of the optimal probe state for the problem $P^\star_{\rm E}(P_{XY}^\prime,N_S^\prime)$.
\end{theorem}
\begin{figure}[H]
    \centering
    \includegraphics[width=\columnwidth]{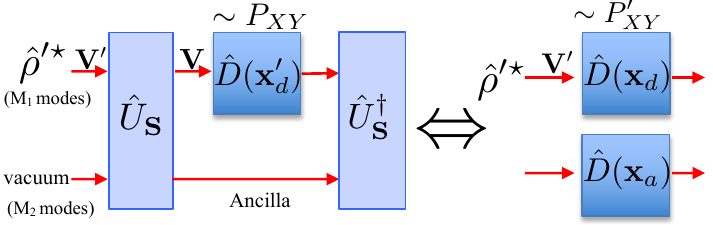}
    \caption{Schematic of the proof of Theorem~\ref{the:app}.
}
    \label{fig: proof_EA}
\end{figure}

\begin{proof}
The proof idea is sketched in Fig.~\ref{fig: proof_EA}.
We first consider introducing $M_2$ vacuum ancilla modes with no displacement acting on them prior to transform. The joint distribution of the displacements on the $(M_1+M_2)$-mode system is
\begin{align}
P_{X^\prime Y}\left(\left\{\bm x_d, \bm x_a\right\}, y\right)=P_{XY}(\bm x_d,y) \delta(\bm x_a),
\end{align}
where $\bm x_d$ is $M_1$-mode displacement on the data and $\bm x_a$ is $M_2$-mode displacement on the ancilla.
Now consider an $(M_1+M_2)$-mode zero-mean Gaussian unitary  $\hat{U}_{\bm{S}}$, which transforms a displacement  $\hat{D}(\bm{\zeta})$ to another displacement~\cite{zhuang2019Scrambling} according to 
\be
U_{\bm{S}}^\dagger \hat{D}(\bm{\zeta}) U_{\bm{S}} =  D (S^{-1} \bm{\zeta}),
\ee    
where $\bm{S} = \begin{pmatrix}
 \bm{S}_{11} & \bm S_{12} \\\bm{S}_{21} & \bm{S}_{22}   
\end{pmatrix}$ has the block form given in the theorem. As shown in Fig.~\ref{fig: proof_EA} left side, via applying the unitary $\hat{U}_{\bm S}$ prior and after the displacement, the data mode has a new distribution
\begin{align}
&P'_{X^\prime Y}\left(\left\{\bm x_d, \bm x_a\right\}, y\right)=P_{X^\prime Y}\left(\bm S\left\{\bm x_d, \bm x_a\right\}, y\right)
\\
&=P_{X^\prime Y}(\{\bm S_{11}\bm x_d+\bm S_{12}\bm x_a, \bm S_{21}\bm x_d+\bm S_{22}\bm x_a\},y) 
\\
&=P_{XY}(\bm S_{11}\bm x_d+\bm S_{12}\bm x_a)\delta(\bm S_{21}\bm x_d+\bm S_{22}\bm x_a).
\end{align}
Note that a symplectic matrix has unity determinant. Now we focus on the displacements on the data mode, which has the marginal distribution
\begin{align}
&P'_{XY}(\bm x_d, y)=\int d^{2M_2} \bm x_a P'_{X^\prime Y}\left(\left\{\bm x_d, \bm x_a\right\}, y\right)
\\
&=\int d^{2M_2} \bm x_a P_{XY}(\bm S_{11}\bm x_d+\bm S_{12}\bm x_a)\delta(\bm S_{21}\bm x_d+\bm S_{22}\bm x_a)
\\
&=\frac{1}{|\bm S_{22}|} P_{XY} (\left(\bm S_{11} - \bm S_{12}\bm S_{22}^{-1} \bm S_{21}\right)\bm x_d,y).
\end{align}

Now suppose we know the optimal input state $\hat{\rho}^\star$ to to the problem $P'_{XY}$ under the energy constraint $N_S^\prime$. We can utilize the probe to obtain an achievable error probability for the distribution $P_{XY}$, as long as we can keep track of the energy change. The input state to the problem $P_{XY}$ is generated by the unitary $\hat{U}_{\bm S}$ applying on the optimal state $\hat\rho^{\prime\star}$ and $M_2$ vacuum ancilla. The input-output relation from $\hat\rho^{\prime\star}$ to the input of $P_{XY}$ forms a Gaussian quantum channel, when tracing out the ancilla. The covariance matrix $\bm V'$ of $\hat\rho^{\prime\star}$ is mapped to the output covariance matrix 
\be 
\bm V= \bm S_{11} \bm V' \bm S_{11}^T+\frac{1}{2}\bm S_{12}\bm S_{12}^T.
\ee 
The corresponding energy
\be 
N_S =\frac{1}{2}\Tr(\bm S_{11} \bm V' \bm S_{11}^T+\frac{1}{2}\bm S_{12}\bm S_{12}^T-M_1),
\ee 
while the original energy 
\be 
N_S^\prime=\frac{1}{2}\Tr(\bm V'-M_1).
\ee 
The achievable error probability for the problem can be lower bounded by the optimal error probability, leading Ineq.~\eqref{general_ineq}.


\end{proof}

We can utilize Theorem~\ref{the:app} to obtain Lemma~\ref{theorem:pe} of the main text as a corollary, which we prove below.

\begin{proof}[Proof of Lemma~\ref{theorem:pe}]
We apply Theorem~\ref{the:app} with $M_1=M_2=M$ modes to complete the proof. 
First, we apply zero-mean symplectic transform pairwise on the $k$-th data mode and $k$-th ancilla mode described by the symplectic matrix,
\be
\bm S_2 = 
\begin{pmatrix}
S_{11} & S_{12}\\  
S_{21} & S_{22}
\end{pmatrix}.
\ee
Here $S_{11}, S_{12}, S_{21}$ and $S_{22}$ are the 2-by-2 block forms.
Then, Eq.~\eqref{PE_transform_full} leads to
\begin{align}
&P'_{XY}(\{\bm x_k\}_{k=1}^M, y)=
\nonumber
\\
&\quad \frac{1}{|S_{22}|^M}P_{XY} \left(\left\{\left(S_{11} - S_{12}S_{22}^{-1} S_{21}\right)\bm x_k\right\}_{k=1}^M,y \right),
\label{eq:2mode-wise}
\end{align}
where we have denoted the data on the $k$-th mode as $\bm x_k$. 
The energy constraints are related via equation similar to Eq.~\eqref{eq:energy_constraint}.

Consider a beamsplitter
$
\bm S_2 = \begin{pmatrix}
\cos(\theta) I_2 & -\sin (\theta)  I_2  \\
\sin(\theta) I_2 & \cos (\theta)  I_2  
\end{pmatrix}.
$
From Eq.~\eqref{eq:2mode-wise}, we have 
\be 
P'_{XY}(\{\bm x_k\}_{k=1}^M, y)=
\frac{1}{\cos^{2M}\theta}P_{XY}\left(\left\{\frac{1}{\cos\theta}\bm x_k\right\}_{k=1}^M, y \right).
\label{PXY,c<1}
\ee 
While one can work out the energy constraints from Eq.~\eqref{eq:energy_constraint}, in this case the Gaussian channel is a mode-wise pure loss channel and we directly have
\be 
N_S^\prime = N_S/|\cos\theta|^2.
\label{NS,c<1}
\ee 
Combining Eq.~\eqref{NS,c<1} and Eq.~\eqref{PXY,c<1}, we have proved the $c<1$ part of the Lemma~\ref{theorem:pe}. (Note that the notation in Lemma~\ref{theorem:pe} has the notations with and without prime exchanged.)

The $c>1$ of Lemma~\ref{theorem:pe} comes from applying a two-mode squeezing operator mode-wise, of which the symplectic map is
$
\bm S_2 = 
\begin{pmatrix}
\cosh r I & \sinh r  Z\\  
\sinh r Z & \cosh r I 
\end{pmatrix}.
$
From Eq.~\eqref{eq:2mode-wise}, we have 
\be 
P'_{XY}(\{\bm x_k\}_{k=1}^M, y)=
\frac{1}{\cosh^{2M}(r)}P_{XY} \left(\left\{\frac{1}{\cosh r }\bm x_k\right\}_{k=1}^M, y \right).
\label{PXY,c>1}
\ee 
While one can work out the energy constraints from Eq.~\eqref{eq:energy_constraint}, in this case the Gaussian channel is a mode-wise quantum-limited amplifier channel and 
\be
N_S= \cosh^2 r N^\prime_S +\cosh^2 r-1.
\label{NS,c>1}
\ee
Combining Eq.~\eqref{NS,c>1} and Eq.~\eqref{PXY,c>1}, we have proved the $c>1$ part of Lemma~\ref{theorem:pe}. (Note that the notation in Lemma~\ref{theorem:pe} has the notations with and without prime exchanged.)

Combining the two cases, Lemma~\ref{theorem:pe} is then proven. 

\end{proof}

\begin{corollary}[Ellipse to circle]
For the following three 2D probability distributions
\begin{align}
P_{XY}(x_1,x_2,y=0) =&P_{XY}^\prime(x_1,x_2,y=0) =  \delta(x_1)\delta(x_2),\\
P_{XY}(x_1,x_2,y=1) =& \frac{1}{C_1} \delta \left( \frac{x_1^2}{a^2}+\frac{x_2^2}{b^2} -1\right),\\
P_{XY}^\prime(x_1,x_2,y=1) =& \frac{1}{C_2} \delta \left( x_1^2+x_2^2 -ab\right)  ,  
\end{align}
we have 
\be
P_{\rm E}^\star(P_{XY}, N_S) \le  P_{\rm E}^\star(P_{XY}^\prime, N_S^\prime ) 
\ee
where
\be
N_S^\prime = N_S - \frac{1}{2} \expval{(a/b-1)\hat{q}^2 + (b/a-1)\hat{p}^2 }.
\ee
\end{corollary}

\begin{proof}
In this case, we have $M_1=2$ data modes and $M_2=0$ ancilla modes. We apply Theorem~\ref{the:app} with the single-mode squeezing symplectic matrix
\be
\bm{S} = 
\begin{pmatrix}
 e^{-r} &0\\
 0& e^{r}
\end{pmatrix}
 = 
 \begin{pmatrix}
\sqrt{a/b} & 0\\
 0& \sqrt{b/a}
\end{pmatrix}
\ee
and the result directly follows from Theorem~\ref{the:app}.
\end{proof}

\begin{corollary}[2d real to 1d complex]
For the following 2D probability distribution functions
\begin{align}
P_{XY}(\bm{x}|y=0) =& P_{XY}^\prime(\bm{x}|y=0) =\delta(x_1) \delta(x_2) \delta(x_3)\delta(x_4) ,\\
P_{XY}(\bm{x}|y=1) =& f(x_1,x_3) \delta(x_2)\delta(x_4),\\
P^\prime_{XY}(\bm{x}|y=1) =& f(x_1,x_2) \delta \left(x_2-x_4\right) \delta \left(x_1+x_3\right) ,  
\end{align}
we have 
\be
P_{\rm E}^\star(P_{XY}, N_S) \le  P_{\rm E}^\star(P_{XY}^\prime, N_S^\prime) 
\ee
where
\be
N_S = N_S^\prime + \frac{1}{2} \expval{\hat{q}_1^2+ \hat{p}_2^2+\{\hat{q}_1,\hat{q}_2\}-\{\hat{p}_1,\hat{p}_2\}}.
\ee
\end{corollary}
\begin{proof}
In this case, we $M_1=2$ data modes and $M_2=0$ ancilla modes. We apply Theorem~\ref{the:app} with the symplectic matrix $\bm S_2 \bm S_1$ in two steps.

First, initiate a phase rotation in the second mode, with symplectic transformation matrix denoted by $S$:
\be
\bm{S}_1 = 
\begin{pmatrix} 
1 & 0 & 0 & 0 \\ 
0 & 1 & 0 & 0 \\ 0 & 0 & 0 & 1 \\ 
0 & 0 & 1 & 0 
\end{pmatrix}
\ee
This matrix preserves $P_{XY}(\cdot|y=0)$, while modifying $P_{XY}(\cdot|y=1)$ to:
\be
P_{XY}^{\prime\prime}(\bm{x}|y=1) = f(x_1,x_4) \delta(x_2)\delta(x_3)
\ee
This transformation maintains the energy unaltered.

Subsequently, the SUM gate  $\hat{U}_{\rm SUM} = e^{-i\hat{q}_1\hat{p}_2 }$ is utilized,
of which the symplectic transformation matrix is  
\be
\bm{S}_2 = 
\begin{pmatrix}
1 & 0 & 0 & 0 \\
0 & 1 & 0 & -1 \\
1 & 0 & 1 & 0 \\
0 & 0 & 0 & 1
\end{pmatrix}.
\ee
This gate does not change $P^{\prime\prime}_{XY}(\cdot|y=0)$ and transform $P_{XY}^{\prime\prime}$ to
\begin{align}
P_{XY}^{\prime}(\bm{x}|y=1) 
=&f(x_1,x_4) \delta(x_2-x4)\delta(x_3+x_1)\\
=&f(x_1,x_2) \delta(x_3+x_1)\delta(x_2-x4).
\end{align}

There is no ancillary mode in the transformation, thus the energy change by Theorem~\ref{the:app} is
\begin{align}
N_S =& N_S^\prime +\frac{1}{2}\operatorname{Tr} \left(\rho \hat{\bm{x}}^T (\bm{S}^T\bm{S}-I) \hat{\bm{x}}\right) \\
=& N_S^\prime +\frac{1}{2}\operatorname{Tr} \left(\rho \hat{\bm{x}}^T 
\begin{pmatrix}
1&0 & 1& 0  \\ 
0&0 & 0& -1\\
1&0&0&0\\
0&-1 & 0& 1\\
\end{pmatrix}
\hat{\bm{x}}\right) \\
=&N_S^\prime+ \frac{1}{2} \expval{\hat{q}_1^2+ \hat{p}_2^2+\{\hat{q}_1,\hat{q}_2\}-\{\hat{p}_1,\hat{p}_2\}}.
\end{align}
\end{proof}






\begin{thebibliography}{60}%
\makeatletter
\providecommand \@ifxundefined [1]{%
 \@ifx{#1\undefined}
}%
\providecommand \@ifnum [1]{%
 \ifnum #1\expandafter \@firstoftwo
 \else \expandafter \@secondoftwo
 \fi
}%
\providecommand \@ifx [1]{%
 \ifx #1\expandafter \@firstoftwo
 \else \expandafter \@secondoftwo
 \fi
}%
\providecommand \natexlab [1]{#1}%
\providecommand \enquote  [1]{``#1''}%
\providecommand \bibnamefont  [1]{#1}%
\providecommand \bibfnamefont [1]{#1}%
\providecommand \citenamefont [1]{#1}%
\providecommand \href@noop [0]{\@secondoftwo}%
\providecommand \href [0]{\begingroup \@sanitize@url \@href}%
\providecommand \@href[1]{\@@startlink{#1}\@@href}%
\providecommand \@@href[1]{\endgroup#1\@@endlink}%
\providecommand \@sanitize@url [0]{\catcode `\\12\catcode `\$12\catcode
  `\&12\catcode `\#12\catcode `\^12\catcode `\_12\catcode `\%12\relax}%
\providecommand \@@startlink[1]{}%
\providecommand \@@endlink[0]{}%
\providecommand \url  [0]{\begingroup\@sanitize@url \@url }%
\providecommand \@url [1]{\endgroup\@href {#1}{\urlprefix }}%
\providecommand \urlprefix  [0]{URL }%
\providecommand \Eprint [0]{\href }%
\providecommand \doibase [0]{https://doi.org/}%
\providecommand \selectlanguage [0]{\@gobble}%
\providecommand \bibinfo  [0]{\@secondoftwo}%
\providecommand \bibfield  [0]{\@secondoftwo}%
\providecommand \translation [1]{[#1]}%
\providecommand \BibitemOpen [0]{}%
\providecommand \bibitemStop [0]{}%
\providecommand \bibitemNoStop [0]{.\EOS\space}%
\providecommand \EOS [0]{\spacefactor3000\relax}%
\providecommand \BibitemShut  [1]{\csname bibitem#1\endcsname}%
\let\auto@bib@innerbib\@empty
\bibitem [{\citenamefont {Weedbrook}\ \emph {et~al.}(2012)\citenamefont
  {Weedbrook}, \citenamefont {Pirandola}, \citenamefont
  {Garc{\'\i}a-Patr{\'o}n}, \citenamefont {Cerf}, \citenamefont {Ralph},
  \citenamefont {Shapiro},\ and\ \citenamefont
  {Lloyd}}]{weedbrook2012gaussian}%
  \BibitemOpen
  \bibfield  {author} {\bibinfo {author} {\bibfnamefont {C.}~\bibnamefont
  {Weedbrook}}, \bibinfo {author} {\bibfnamefont {S.}~\bibnamefont
  {Pirandola}}, \bibinfo {author} {\bibfnamefont {R.}~\bibnamefont
  {Garc{\'\i}a-Patr{\'o}n}}, \bibinfo {author} {\bibfnamefont {N.~J.}\
  \bibnamefont {Cerf}}, \bibinfo {author} {\bibfnamefont {T.~C.}\ \bibnamefont
  {Ralph}}, \bibinfo {author} {\bibfnamefont {J.~H.}\ \bibnamefont {Shapiro}},\
  and\ \bibinfo {author} {\bibfnamefont {S.}~\bibnamefont {Lloyd}},\ }\bibfield
   {title} {\bibinfo {title} {Gaussian quantum information},\ }\href@noop {}
  {\bibfield  {journal} {\bibinfo  {journal} {Rev. Mod. Phys.}\ }\textbf
  {\bibinfo {volume} {84}},\ \bibinfo {pages} {621} (\bibinfo {year}
  {2012})}\BibitemShut {NoStop}%
\bibitem [{\citenamefont {Caves}(1981)}]{caves1981quantum}%
  \BibitemOpen
  \bibfield  {author} {\bibinfo {author} {\bibfnamefont {C.~M.}\ \bibnamefont
  {Caves}},\ }\bibfield  {title} {\bibinfo {title} {Quantum-mechanical noise in
  an interferometer},\ }\href@noop {} {\bibfield  {journal} {\bibinfo
  {journal} {Phys. Rev. D}\ }\textbf {\bibinfo {volume} {23}},\ \bibinfo
  {pages} {1693} (\bibinfo {year} {1981})}\BibitemShut {NoStop}%
\bibitem [{\citenamefont {Ganapathy}\ \emph {et~al.}(2023)\citenamefont
  {Ganapathy}, \citenamefont {Jia}, \citenamefont {Nakano}, \citenamefont {Xu},
  \citenamefont {Aritomi}, \citenamefont {Cullen}, \citenamefont {Kijbunchoo},
  \citenamefont {Dwyer}, \citenamefont {Mullavey},\ and\ \citenamefont
  {McCuller}}]{ganapathy2023broadband}%
  \BibitemOpen
  \bibfield  {author} {\bibinfo {author} {\bibfnamefont {D.}~\bibnamefont
  {Ganapathy}}, \bibinfo {author} {\bibfnamefont {W.}~\bibnamefont {Jia}},
  \bibinfo {author} {\bibfnamefont {M.}~\bibnamefont {Nakano}}, \bibinfo
  {author} {\bibfnamefont {V.}~\bibnamefont {Xu}}, \bibinfo {author}
  {\bibfnamefont {N.}~\bibnamefont {Aritomi}}, \bibinfo {author} {\bibfnamefont
  {T.}~\bibnamefont {Cullen}}, \bibinfo {author} {\bibfnamefont
  {N.}~\bibnamefont {Kijbunchoo}}, \bibinfo {author} {\bibfnamefont {S.~E.}\
  \bibnamefont {Dwyer}}, \bibinfo {author} {\bibfnamefont {A.}~\bibnamefont
  {Mullavey}},\ and\ \bibinfo {author} {\bibfnamefont {L.~e.~a.}\ \bibnamefont
  {McCuller}} (\bibinfo {collaboration} {LIGO O4 Detector Collaboration}),\
  }\bibfield  {title} {\bibinfo {title} {Broadband quantum enhancement of the
  ligo detectors with frequency-dependent squeezing},\ }\href
  {https://doi.org/10.1103/PhysRevX.13.041021} {\bibfield  {journal} {\bibinfo
  {journal} {Phys. Rev. X}\ }\textbf {\bibinfo {volume} {13}},\ \bibinfo
  {pages} {041021} (\bibinfo {year} {2023})}\BibitemShut {NoStop}%
\bibitem [{\citenamefont {Demkowicz-Dobrza{\'n}ski}\ \emph
  {et~al.}(2013)\citenamefont {Demkowicz-Dobrza{\'n}ski}, \citenamefont
  {Banaszek},\ and\ \citenamefont {Schnabel}}]{demkowicz2013fundamental}%
  \BibitemOpen
  \bibfield  {author} {\bibinfo {author} {\bibfnamefont {R.}~\bibnamefont
  {Demkowicz-Dobrza{\'n}ski}}, \bibinfo {author} {\bibfnamefont
  {K.}~\bibnamefont {Banaszek}},\ and\ \bibinfo {author} {\bibfnamefont
  {R.}~\bibnamefont {Schnabel}},\ }\bibfield  {title} {\bibinfo {title}
  {Fundamental quantum interferometry bound for the squeezed-light-enhanced
  gravitational wave detector geo 600},\ }\href@noop {} {\bibfield  {journal}
  {\bibinfo  {journal} {Phys. Rev. A}\ }\textbf {\bibinfo {volume} {88}},\
  \bibinfo {pages} {041802} (\bibinfo {year} {2013})}\BibitemShut {NoStop}%
\bibitem [{\citenamefont {Escher}\ \emph {et~al.}(2011)\citenamefont {Escher},
  \citenamefont {de~Matos~Filho},\ and\ \citenamefont
  {Davidovich}}]{escher2011general}%
  \BibitemOpen
  \bibfield  {author} {\bibinfo {author} {\bibfnamefont {B.}~\bibnamefont
  {Escher}}, \bibinfo {author} {\bibfnamefont {R.~L.}\ \bibnamefont
  {de~Matos~Filho}},\ and\ \bibinfo {author} {\bibfnamefont {L.}~\bibnamefont
  {Davidovich}},\ }\bibfield  {title} {\bibinfo {title} {General framework for
  estimating the ultimate precision limit in noisy quantum-enhanced
  metrology},\ }\href@noop {} {\bibfield  {journal} {\bibinfo  {journal} {Nat.
  Phys.}\ }\textbf {\bibinfo {volume} {7}},\ \bibinfo {pages} {406} (\bibinfo
  {year} {2011})}\BibitemShut {NoStop}%
\bibitem [{\citenamefont {Zhuang}\ \emph {et~al.}(2018)\citenamefont {Zhuang},
  \citenamefont {Zhang},\ and\ \citenamefont
  {Shapiro}}]{zhuang2018distributed}%
  \BibitemOpen
  \bibfield  {author} {\bibinfo {author} {\bibfnamefont {Q.}~\bibnamefont
  {Zhuang}}, \bibinfo {author} {\bibfnamefont {Z.}~\bibnamefont {Zhang}},\ and\
  \bibinfo {author} {\bibfnamefont {J.~H.}\ \bibnamefont {Shapiro}},\
  }\bibfield  {title} {\bibinfo {title} {Distributed quantum sensing using
  continuous-variable multipartite entanglement},\ }\href@noop {} {\bibfield
  {journal} {\bibinfo  {journal} {Phys. Rev. A}\ }\textbf {\bibinfo {volume}
  {97}},\ \bibinfo {pages} {032329} (\bibinfo {year} {2018})}\BibitemShut
  {NoStop}%
\bibitem [{\citenamefont {Ge}\ \emph {et~al.}(2018)\citenamefont {Ge},
  \citenamefont {Jacobs}, \citenamefont {Eldredge}, \citenamefont {Gorshkov},\
  and\ \citenamefont {Foss-Feig}}]{ge2018distributed}%
  \BibitemOpen
  \bibfield  {author} {\bibinfo {author} {\bibfnamefont {W.}~\bibnamefont
  {Ge}}, \bibinfo {author} {\bibfnamefont {K.}~\bibnamefont {Jacobs}}, \bibinfo
  {author} {\bibfnamefont {Z.}~\bibnamefont {Eldredge}}, \bibinfo {author}
  {\bibfnamefont {A.~V.}\ \bibnamefont {Gorshkov}},\ and\ \bibinfo {author}
  {\bibfnamefont {M.}~\bibnamefont {Foss-Feig}},\ }\bibfield  {title} {\bibinfo
  {title} {Distributed quantum metrology with linear networks and separable
  inputs},\ }\href@noop {} {\bibfield  {journal} {\bibinfo  {journal} {Phys.
  Rev. Lett.}\ }\textbf {\bibinfo {volume} {121}},\ \bibinfo {pages} {043604}
  (\bibinfo {year} {2018})}\BibitemShut {NoStop}%
\bibitem [{\citenamefont {Zhang}\ and\ \citenamefont
  {Zhuang}(2021)}]{zhang2021distributed}%
  \BibitemOpen
  \bibfield  {author} {\bibinfo {author} {\bibfnamefont {Z.}~\bibnamefont
  {Zhang}}\ and\ \bibinfo {author} {\bibfnamefont {Q.}~\bibnamefont {Zhuang}},\
  }\bibfield  {title} {\bibinfo {title} {Distributed quantum sensing},\
  }\href@noop {} {\bibfield  {journal} {\bibinfo  {journal} {Quantum Sci.
  Techno.}\ }\textbf {\bibinfo {volume} {6}},\ \bibinfo {pages} {043001}
  (\bibinfo {year} {2021})}\BibitemShut {NoStop}%
\bibitem [{\citenamefont {Guo}\ \emph {et~al.}(2020)\citenamefont {Guo},
  \citenamefont {Breum}, \citenamefont {Borregaard}, \citenamefont {Izumi},
  \citenamefont {Larsen}, \citenamefont {Gehring}, \citenamefont {Christandl},
  \citenamefont {Neergaard-Nielsen},\ and\ \citenamefont
  {Andersen}}]{guo2020distributed}%
  \BibitemOpen
  \bibfield  {author} {\bibinfo {author} {\bibfnamefont {X.}~\bibnamefont
  {Guo}}, \bibinfo {author} {\bibfnamefont {C.~R.}\ \bibnamefont {Breum}},
  \bibinfo {author} {\bibfnamefont {J.}~\bibnamefont {Borregaard}}, \bibinfo
  {author} {\bibfnamefont {S.}~\bibnamefont {Izumi}}, \bibinfo {author}
  {\bibfnamefont {M.~V.}\ \bibnamefont {Larsen}}, \bibinfo {author}
  {\bibfnamefont {T.}~\bibnamefont {Gehring}}, \bibinfo {author} {\bibfnamefont
  {M.}~\bibnamefont {Christandl}}, \bibinfo {author} {\bibfnamefont {J.~S.}\
  \bibnamefont {Neergaard-Nielsen}},\ and\ \bibinfo {author} {\bibfnamefont
  {U.~L.}\ \bibnamefont {Andersen}},\ }\bibfield  {title} {\bibinfo {title}
  {Distributed quantum sensing in a continuous-variable entangled network},\
  }\href@noop {} {\bibfield  {journal} {\bibinfo  {journal} {Nat. Phys.}\
  }\textbf {\bibinfo {volume} {16}},\ \bibinfo {pages} {281} (\bibinfo {year}
  {2020})}\BibitemShut {NoStop}%
\bibitem [{\citenamefont {Liu}\ \emph {et~al.}(2021{\natexlab{a}})\citenamefont
  {Liu}, \citenamefont {Zhang}, \citenamefont {Li}, \citenamefont {Zhang},
  \citenamefont {Yin}, \citenamefont {Fei}, \citenamefont {Li}, \citenamefont
  {Liu}, \citenamefont {Xu}, \citenamefont {Chen} \emph
  {et~al.}}]{liu2021distributed}%
  \BibitemOpen
  \bibfield  {author} {\bibinfo {author} {\bibfnamefont {L.-Z.}\ \bibnamefont
  {Liu}}, \bibinfo {author} {\bibfnamefont {Y.-Z.}\ \bibnamefont {Zhang}},
  \bibinfo {author} {\bibfnamefont {Z.-D.}\ \bibnamefont {Li}}, \bibinfo
  {author} {\bibfnamefont {R.}~\bibnamefont {Zhang}}, \bibinfo {author}
  {\bibfnamefont {X.-F.}\ \bibnamefont {Yin}}, \bibinfo {author} {\bibfnamefont
  {Y.-Y.}\ \bibnamefont {Fei}}, \bibinfo {author} {\bibfnamefont
  {L.}~\bibnamefont {Li}}, \bibinfo {author} {\bibfnamefont {N.-L.}\
  \bibnamefont {Liu}}, \bibinfo {author} {\bibfnamefont {F.}~\bibnamefont
  {Xu}}, \bibinfo {author} {\bibfnamefont {Y.-A.}\ \bibnamefont {Chen}}, \emph
  {et~al.},\ }\bibfield  {title} {\bibinfo {title} {Distributed quantum phase
  estimation with entangled photons},\ }\href@noop {} {\bibfield  {journal}
  {\bibinfo  {journal} {Nat. Photon.}\ }\textbf {\bibinfo {volume} {15}},\
  \bibinfo {pages} {137} (\bibinfo {year} {2021}{\natexlab{a}})}\BibitemShut
  {NoStop}%
\bibitem [{\citenamefont {Xia}\ \emph {et~al.}(2020)\citenamefont {Xia},
  \citenamefont {Li}, \citenamefont {Clark}, \citenamefont {Hart},
  \citenamefont {Zhuang},\ and\ \citenamefont {Zhang}}]{xia2020demonstration}%
  \BibitemOpen
  \bibfield  {author} {\bibinfo {author} {\bibfnamefont {Y.}~\bibnamefont
  {Xia}}, \bibinfo {author} {\bibfnamefont {W.}~\bibnamefont {Li}}, \bibinfo
  {author} {\bibfnamefont {W.}~\bibnamefont {Clark}}, \bibinfo {author}
  {\bibfnamefont {D.}~\bibnamefont {Hart}}, \bibinfo {author} {\bibfnamefont
  {Q.}~\bibnamefont {Zhuang}},\ and\ \bibinfo {author} {\bibfnamefont
  {Z.}~\bibnamefont {Zhang}},\ }\bibfield  {title} {\bibinfo {title}
  {Demonstration of a reconfigurable entangled radio-frequency photonic sensor
  network},\ }\href {https://doi.org/10.1103/PhysRevLett.124.150502} {\bibfield
   {journal} {\bibinfo  {journal} {Phys. Rev. Lett.}\ }\textbf {\bibinfo
  {volume} {124}},\ \bibinfo {pages} {150502} (\bibinfo {year}
  {2020})}\BibitemShut {NoStop}%
\bibitem [{\citenamefont {Xia}\ \emph {et~al.}(2023)\citenamefont {Xia},
  \citenamefont {Agrawal}, \citenamefont {Pluchar}, \citenamefont {Brady},
  \citenamefont {Liu}, \citenamefont {Zhuang}, \citenamefont {Wilson},\ and\
  \citenamefont {Zhang}}]{xia2023entanglement}%
  \BibitemOpen
  \bibfield  {author} {\bibinfo {author} {\bibfnamefont {Y.}~\bibnamefont
  {Xia}}, \bibinfo {author} {\bibfnamefont {A.~R.}\ \bibnamefont {Agrawal}},
  \bibinfo {author} {\bibfnamefont {C.~M.}\ \bibnamefont {Pluchar}}, \bibinfo
  {author} {\bibfnamefont {A.~J.}\ \bibnamefont {Brady}}, \bibinfo {author}
  {\bibfnamefont {Z.}~\bibnamefont {Liu}}, \bibinfo {author} {\bibfnamefont
  {Q.}~\bibnamefont {Zhuang}}, \bibinfo {author} {\bibfnamefont {D.~J.}\
  \bibnamefont {Wilson}},\ and\ \bibinfo {author} {\bibfnamefont
  {Z.}~\bibnamefont {Zhang}},\ }\bibfield  {title} {\bibinfo {title}
  {Entanglement-enhanced optomechanical sensing},\ }\href
  {https://doi.org/10.1038/s41566-023-01178-0} {\bibfield  {journal} {\bibinfo
  {journal} {Nat. Photon.}\ }\textbf {\bibinfo {volume} {17}},\ \bibinfo
  {pages} {470–477} (\bibinfo {year} {2023})}\BibitemShut {NoStop}%
\bibitem [{\citenamefont {Malia}\ \emph {et~al.}(2022)\citenamefont {Malia},
  \citenamefont {Wu}, \citenamefont {Mart{\'\i}nez-Rinc{\'o}n},\ and\
  \citenamefont {Kasevich}}]{malia2022distributed}%
  \BibitemOpen
  \bibfield  {author} {\bibinfo {author} {\bibfnamefont {B.~K.}\ \bibnamefont
  {Malia}}, \bibinfo {author} {\bibfnamefont {Y.}~\bibnamefont {Wu}}, \bibinfo
  {author} {\bibfnamefont {J.}~\bibnamefont {Mart{\'\i}nez-Rinc{\'o}n}},\ and\
  \bibinfo {author} {\bibfnamefont {M.~A.}\ \bibnamefont {Kasevich}},\
  }\bibfield  {title} {\bibinfo {title} {Distributed quantum sensing with
  mode-entangled spin-squeezed atomic states},\ }\href@noop {} {\bibfield
  {journal} {\bibinfo  {journal} {Nature}\ }\textbf {\bibinfo {volume} {612}},\
  \bibinfo {pages} {661} (\bibinfo {year} {2022})}\BibitemShut {NoStop}%
\bibitem [{\citenamefont {Zhuang}\ and\ \citenamefont
  {Zhang}(2019)}]{zhuang2019physical}%
  \BibitemOpen
  \bibfield  {author} {\bibinfo {author} {\bibfnamefont {Q.}~\bibnamefont
  {Zhuang}}\ and\ \bibinfo {author} {\bibfnamefont {Z.}~\bibnamefont {Zhang}},\
  }\bibfield  {title} {\bibinfo {title} {Physical-layer supervised learning
  assisted by an entangled sensor network},\ }\href
  {https://doi.org/10.1103/PhysRevX.9.041023} {\bibfield  {journal} {\bibinfo
  {journal} {Phys. Rev. X}\ }\textbf {\bibinfo {volume} {9}},\ \bibinfo {pages}
  {041023} (\bibinfo {year} {2019})}\BibitemShut {NoStop}%
\bibitem [{\citenamefont {Cerezo}\ \emph {et~al.}(2021)\citenamefont {Cerezo},
  \citenamefont {Arrasmith}, \citenamefont {Babbush}, \citenamefont {Benjamin},
  \citenamefont {Endo}, \citenamefont {Fujii}, \citenamefont {McClean},
  \citenamefont {Mitarai}, \citenamefont {Yuan}, \citenamefont {Cincio} \emph
  {et~al.}}]{cerezo2021variational}%
  \BibitemOpen
  \bibfield  {author} {\bibinfo {author} {\bibfnamefont {M.}~\bibnamefont
  {Cerezo}}, \bibinfo {author} {\bibfnamefont {A.}~\bibnamefont {Arrasmith}},
  \bibinfo {author} {\bibfnamefont {R.}~\bibnamefont {Babbush}}, \bibinfo
  {author} {\bibfnamefont {S.~C.}\ \bibnamefont {Benjamin}}, \bibinfo {author}
  {\bibfnamefont {S.}~\bibnamefont {Endo}}, \bibinfo {author} {\bibfnamefont
  {K.}~\bibnamefont {Fujii}}, \bibinfo {author} {\bibfnamefont {J.~R.}\
  \bibnamefont {McClean}}, \bibinfo {author} {\bibfnamefont {K.}~\bibnamefont
  {Mitarai}}, \bibinfo {author} {\bibfnamefont {X.}~\bibnamefont {Yuan}},
  \bibinfo {author} {\bibfnamefont {L.}~\bibnamefont {Cincio}}, \emph
  {et~al.},\ }\bibfield  {title} {\bibinfo {title} {Variational quantum
  algorithms},\ }\href@noop {} {\bibfield  {journal} {\bibinfo  {journal} {Nat.
  Rev. Phys.}\ }\textbf {\bibinfo {volume} {3}},\ \bibinfo {pages} {625}
  (\bibinfo {year} {2021})}\BibitemShut {NoStop}%
\bibitem [{\citenamefont {Xia}\ \emph {et~al.}(2021)\citenamefont {Xia},
  \citenamefont {Li}, \citenamefont {Zhuang},\ and\ \citenamefont
  {Zhang}}]{xia2021quantum}%
  \BibitemOpen
  \bibfield  {author} {\bibinfo {author} {\bibfnamefont {Y.}~\bibnamefont
  {Xia}}, \bibinfo {author} {\bibfnamefont {W.}~\bibnamefont {Li}}, \bibinfo
  {author} {\bibfnamefont {Q.}~\bibnamefont {Zhuang}},\ and\ \bibinfo {author}
  {\bibfnamefont {Z.}~\bibnamefont {Zhang}},\ }\bibfield  {title} {\bibinfo
  {title} {Quantum-enhanced data classification with a variational entangled
  sensor network},\ }\href {https://doi.org/10.1103/PhysRevX.11.021047}
  {\bibfield  {journal} {\bibinfo  {journal} {Phys. Rev. X}\ }\textbf {\bibinfo
  {volume} {11}},\ \bibinfo {pages} {021047} (\bibinfo {year}
  {2021})}\BibitemShut {NoStop}%
\bibitem [{\citenamefont {Eickbusch}\ \emph {et~al.}(2022)\citenamefont
  {Eickbusch}, \citenamefont {Sivak}, \citenamefont {Ding}, \citenamefont
  {Elder}, \citenamefont {Jha}, \citenamefont {Venkatraman}, \citenamefont
  {Royer}, \citenamefont {Girvin}, \citenamefont {Schoelkopf},\ and\
  \citenamefont {Devoret}}]{eickbusch2022fast}%
  \BibitemOpen
  \bibfield  {author} {\bibinfo {author} {\bibfnamefont {A.}~\bibnamefont
  {Eickbusch}}, \bibinfo {author} {\bibfnamefont {V.}~\bibnamefont {Sivak}},
  \bibinfo {author} {\bibfnamefont {A.~Z.}\ \bibnamefont {Ding}}, \bibinfo
  {author} {\bibfnamefont {S.~S.}\ \bibnamefont {Elder}}, \bibinfo {author}
  {\bibfnamefont {S.~R.}\ \bibnamefont {Jha}}, \bibinfo {author} {\bibfnamefont
  {J.}~\bibnamefont {Venkatraman}}, \bibinfo {author} {\bibfnamefont
  {B.}~\bibnamefont {Royer}}, \bibinfo {author} {\bibfnamefont
  {S.}~\bibnamefont {Girvin}}, \bibinfo {author} {\bibfnamefont {R.~J.}\
  \bibnamefont {Schoelkopf}},\ and\ \bibinfo {author} {\bibfnamefont {M.~H.}\
  \bibnamefont {Devoret}},\ }\bibfield  {title} {\bibinfo {title} {Fast
  universal control of an oscillator with weak dispersive coupling to a
  qubit},\ }\href@noop {} {\bibfield  {journal} {\bibinfo  {journal} {Nat.
  Phys.}\ }\textbf {\bibinfo {volume} {18}},\ \bibinfo {pages} {1464–1469}
  (\bibinfo {year} {2022})}\BibitemShut {NoStop}%
\bibitem [{\citenamefont {Diringer}\ \emph {et~al.}(2023)\citenamefont
  {Diringer}, \citenamefont {Blumenthal}, \citenamefont {Grinberg},
  \citenamefont {Jiang},\ and\ \citenamefont
  {Hacohen-Gourgy}}]{diringer2023conditional}%
  \BibitemOpen
  \bibfield  {author} {\bibinfo {author} {\bibfnamefont {A.~A.}\ \bibnamefont
  {Diringer}}, \bibinfo {author} {\bibfnamefont {E.}~\bibnamefont
  {Blumenthal}}, \bibinfo {author} {\bibfnamefont {A.}~\bibnamefont
  {Grinberg}}, \bibinfo {author} {\bibfnamefont {L.}~\bibnamefont {Jiang}},\
  and\ \bibinfo {author} {\bibfnamefont {S.}~\bibnamefont {Hacohen-Gourgy}},\
  }\bibfield  {title} {\bibinfo {title} {Conditional not displacement: fast
  multi-oscillator control with a single qubit},\ }\href@noop {} {\bibfield
  {journal} {\bibinfo  {journal} {arXiv:2301.09831}\ } (\bibinfo {year}
  {2023})}\BibitemShut {NoStop}%
\bibitem [{\citenamefont {Fl{\"u}hmann}\ \emph {et~al.}(2019)\citenamefont
  {Fl{\"u}hmann}, \citenamefont {Nguyen}, \citenamefont {Marinelli},
  \citenamefont {Negnevitsky}, \citenamefont {Mehta},\ and\ \citenamefont
  {Home}}]{fluhmann2019encoding}%
  \BibitemOpen
  \bibfield  {author} {\bibinfo {author} {\bibfnamefont {C.}~\bibnamefont
  {Fl{\"u}hmann}}, \bibinfo {author} {\bibfnamefont {T.~L.}\ \bibnamefont
  {Nguyen}}, \bibinfo {author} {\bibfnamefont {M.}~\bibnamefont {Marinelli}},
  \bibinfo {author} {\bibfnamefont {V.}~\bibnamefont {Negnevitsky}}, \bibinfo
  {author} {\bibfnamefont {K.}~\bibnamefont {Mehta}},\ and\ \bibinfo {author}
  {\bibfnamefont {J.}~\bibnamefont {Home}},\ }\bibfield  {title} {\bibinfo
  {title} {Encoding a qubit in a trapped-ion mechanical oscillator},\
  }\href@noop {} {\bibfield  {journal} {\bibinfo  {journal} {Nature}\ }\textbf
  {\bibinfo {volume} {566}},\ \bibinfo {pages} {513} (\bibinfo {year}
  {2019})}\BibitemShut {NoStop}%
\bibitem [{\citenamefont {Hauer}\ \emph {et~al.}(2023)\citenamefont {Hauer},
  \citenamefont {Combes},\ and\ \citenamefont {Teufel}}]{hauer2023nonlinear}%
  \BibitemOpen
  \bibfield  {author} {\bibinfo {author} {\bibfnamefont {B.~D.}\ \bibnamefont
  {Hauer}}, \bibinfo {author} {\bibfnamefont {J.}~\bibnamefont {Combes}},\ and\
  \bibinfo {author} {\bibfnamefont {J.~D.}\ \bibnamefont {Teufel}},\ }\bibfield
   {title} {\bibinfo {title} {Nonlinear sideband cooling to a cat state of
  motion},\ }\href@noop {} {\bibfield  {journal} {\bibinfo  {journal} {Phys.
  Rev. Lett.}\ }\textbf {\bibinfo {volume} {130}},\ \bibinfo {pages} {213604}
  (\bibinfo {year} {2023})}\BibitemShut {NoStop}%
\bibitem [{\citenamefont {Gerashchenko}\ \emph {et~al.}(2024)\citenamefont
  {Gerashchenko}, \citenamefont {Najera~Santos}, \citenamefont {Rousseau},
  \citenamefont {Patange}, \citenamefont {Riva}, \citenamefont {Villiers},
  \citenamefont {Briant}, \citenamefont {Cohadon}, \citenamefont {Heidmann},
  \citenamefont {Le~Sueur} \emph {et~al.}}]{gerashchenko2024towards}%
  \BibitemOpen
  \bibfield  {author} {\bibinfo {author} {\bibfnamefont {K.}~\bibnamefont
  {Gerashchenko}}, \bibinfo {author} {\bibfnamefont {B.~L.}\ \bibnamefont
  {Najera~Santos}}, \bibinfo {author} {\bibfnamefont {R.}~\bibnamefont
  {Rousseau}}, \bibinfo {author} {\bibfnamefont {H.}~\bibnamefont {Patange}},
  \bibinfo {author} {\bibfnamefont {A.}~\bibnamefont {Riva}}, \bibinfo {author}
  {\bibfnamefont {M.}~\bibnamefont {Villiers}}, \bibinfo {author}
  {\bibfnamefont {T.}~\bibnamefont {Briant}}, \bibinfo {author} {\bibfnamefont
  {P.-F.}\ \bibnamefont {Cohadon}}, \bibinfo {author} {\bibfnamefont
  {A.}~\bibnamefont {Heidmann}}, \bibinfo {author} {\bibfnamefont
  {H.}~\bibnamefont {Le~Sueur}}, \emph {et~al.},\ }\bibfield  {title} {\bibinfo
  {title} {Towards quantum control of an ultracoherent mechanical resonator
  with a fluxonium qubit},\ }\href@noop {} {\bibfield  {journal} {\bibinfo
  {journal} {Bulletin of the American Physical Society}\ } (\bibinfo {year}
  {2024})}\BibitemShut {NoStop}%
\bibitem [{\citenamefont {Backes}\ \emph
  {et~al.}(2021{\natexlab{a}})\citenamefont {Backes}, \citenamefont {Palken},
  \citenamefont {Kenany}, \citenamefont {Brubaker}, \citenamefont {Cahn},
  \citenamefont {Droster}, \citenamefont {Hilton}, \citenamefont {Ghosh},
  \citenamefont {Jackson}, \citenamefont {Lamoreaux} \emph
  {et~al.}}]{backes2021quantum}%
  \BibitemOpen
  \bibfield  {author} {\bibinfo {author} {\bibfnamefont {K.~M.}\ \bibnamefont
  {Backes}}, \bibinfo {author} {\bibfnamefont {D.~A.}\ \bibnamefont {Palken}},
  \bibinfo {author} {\bibfnamefont {S.~A.}\ \bibnamefont {Kenany}}, \bibinfo
  {author} {\bibfnamefont {B.~M.}\ \bibnamefont {Brubaker}}, \bibinfo {author}
  {\bibfnamefont {S.}~\bibnamefont {Cahn}}, \bibinfo {author} {\bibfnamefont
  {A.}~\bibnamefont {Droster}}, \bibinfo {author} {\bibfnamefont {G.~C.}\
  \bibnamefont {Hilton}}, \bibinfo {author} {\bibfnamefont {S.}~\bibnamefont
  {Ghosh}}, \bibinfo {author} {\bibfnamefont {H.}~\bibnamefont {Jackson}},
  \bibinfo {author} {\bibfnamefont {S.~K.}\ \bibnamefont {Lamoreaux}}, \emph
  {et~al.},\ }\bibfield  {title} {\bibinfo {title} {A quantum enhanced search
  for dark matter axions},\ }\href@noop {} {\bibfield  {journal} {\bibinfo
  {journal} {Nature}\ }\textbf {\bibinfo {volume} {590}},\ \bibinfo {pages}
  {238} (\bibinfo {year} {2021}{\natexlab{a}})}\BibitemShut {NoStop}%
\bibitem [{\citenamefont {Shi}\ and\ \citenamefont
  {Zhuang}(2023)}]{shi2023ultimate}%
  \BibitemOpen
  \bibfield  {author} {\bibinfo {author} {\bibfnamefont {H.}~\bibnamefont
  {Shi}}\ and\ \bibinfo {author} {\bibfnamefont {Q.}~\bibnamefont {Zhuang}},\
  }\bibfield  {title} {\bibinfo {title} {Ultimate precision limit of noise
  sensing and dark matter search},\ }\href@noop {} {\bibfield  {journal}
  {\bibinfo  {journal} {npj Quantum Inf.}\ }\textbf {\bibinfo {volume} {9}},\
  \bibinfo {pages} {27} (\bibinfo {year} {2023})}\BibitemShut {NoStop}%
\bibitem [{\citenamefont {Brady}\ \emph {et~al.}(2022)\citenamefont {Brady},
  \citenamefont {Gao}, \citenamefont {Harnik}, \citenamefont {Liu},
  \citenamefont {Zhang},\ and\ \citenamefont {Zhuang}}]{brady2022entangled}%
  \BibitemOpen
  \bibfield  {author} {\bibinfo {author} {\bibfnamefont {A.~J.}\ \bibnamefont
  {Brady}}, \bibinfo {author} {\bibfnamefont {C.}~\bibnamefont {Gao}}, \bibinfo
  {author} {\bibfnamefont {R.}~\bibnamefont {Harnik}}, \bibinfo {author}
  {\bibfnamefont {Z.}~\bibnamefont {Liu}}, \bibinfo {author} {\bibfnamefont
  {Z.}~\bibnamefont {Zhang}},\ and\ \bibinfo {author} {\bibfnamefont
  {Q.}~\bibnamefont {Zhuang}},\ }\bibfield  {title} {\bibinfo {title}
  {Entangled sensor-networks for dark-matter searches},\ }\href
  {https://doi.org/10.1103/PRXQuantum.3.030333} {\bibfield  {journal} {\bibinfo
   {journal} {PRX Quantum}\ }\textbf {\bibinfo {volume} {3}},\ \bibinfo {pages}
  {030333} (\bibinfo {year} {2022})}\BibitemShut {NoStop}%
\bibitem [{\citenamefont {Huang}\ \emph {et~al.}(2022)\citenamefont {Huang},
  \citenamefont {Broughton}, \citenamefont {Cotler}, \citenamefont {Chen},
  \citenamefont {Li}, \citenamefont {Mohseni}, \citenamefont {Neven},
  \citenamefont {Babbush}, \citenamefont {Kueng}, \citenamefont {Preskill}
  \emph {et~al.}}]{huang2022quantum}%
  \BibitemOpen
  \bibfield  {author} {\bibinfo {author} {\bibfnamefont {H.-Y.}\ \bibnamefont
  {Huang}}, \bibinfo {author} {\bibfnamefont {M.}~\bibnamefont {Broughton}},
  \bibinfo {author} {\bibfnamefont {J.}~\bibnamefont {Cotler}}, \bibinfo
  {author} {\bibfnamefont {S.}~\bibnamefont {Chen}}, \bibinfo {author}
  {\bibfnamefont {J.}~\bibnamefont {Li}}, \bibinfo {author} {\bibfnamefont
  {M.}~\bibnamefont {Mohseni}}, \bibinfo {author} {\bibfnamefont
  {H.}~\bibnamefont {Neven}}, \bibinfo {author} {\bibfnamefont
  {R.}~\bibnamefont {Babbush}}, \bibinfo {author} {\bibfnamefont
  {R.}~\bibnamefont {Kueng}}, \bibinfo {author} {\bibfnamefont
  {J.}~\bibnamefont {Preskill}}, \emph {et~al.},\ }\bibfield  {title} {\bibinfo
  {title} {Quantum advantage in learning from experiments},\ }\href@noop {}
  {\bibfield  {journal} {\bibinfo  {journal} {Science}\ }\textbf {\bibinfo
  {volume} {376}},\ \bibinfo {pages} {1182} (\bibinfo {year}
  {2022})}\BibitemShut {NoStop}%
\bibitem [{\citenamefont {Banchi}\ \emph {et~al.}(2021)\citenamefont {Banchi},
  \citenamefont {Pereira},\ and\ \citenamefont
  {Pirandola}}]{banchi2021generalization}%
  \BibitemOpen
  \bibfield  {author} {\bibinfo {author} {\bibfnamefont {L.}~\bibnamefont
  {Banchi}}, \bibinfo {author} {\bibfnamefont {J.}~\bibnamefont {Pereira}},\
  and\ \bibinfo {author} {\bibfnamefont {S.}~\bibnamefont {Pirandola}},\
  }\bibfield  {title} {\bibinfo {title} {Generalization in quantum machine
  learning: A quantum information standpoint},\ }\href@noop {} {\bibfield
  {journal} {\bibinfo  {journal} {PRX Quantum}\ }\textbf {\bibinfo {volume}
  {2}},\ \bibinfo {pages} {040321} (\bibinfo {year} {2021})}\BibitemShut
  {NoStop}%
\bibitem [{\citenamefont {Caro}\ \emph {et~al.}(2023)\citenamefont {Caro},
  \citenamefont {Huang}, \citenamefont {Ezzell}, \citenamefont {Gibbs},
  \citenamefont {Sornborger}, \citenamefont {Cincio}, \citenamefont {Coles},\
  and\ \citenamefont {Holmes}}]{caro2023out}%
  \BibitemOpen
  \bibfield  {author} {\bibinfo {author} {\bibfnamefont {M.~C.}\ \bibnamefont
  {Caro}}, \bibinfo {author} {\bibfnamefont {H.-Y.}\ \bibnamefont {Huang}},
  \bibinfo {author} {\bibfnamefont {N.}~\bibnamefont {Ezzell}}, \bibinfo
  {author} {\bibfnamefont {J.}~\bibnamefont {Gibbs}}, \bibinfo {author}
  {\bibfnamefont {A.~T.}\ \bibnamefont {Sornborger}}, \bibinfo {author}
  {\bibfnamefont {L.}~\bibnamefont {Cincio}}, \bibinfo {author} {\bibfnamefont
  {P.~J.}\ \bibnamefont {Coles}},\ and\ \bibinfo {author} {\bibfnamefont
  {Z.}~\bibnamefont {Holmes}},\ }\bibfield  {title} {\bibinfo {title}
  {Out-of-distribution generalization for learning quantum dynamics},\
  }\href@noop {} {\bibfield  {journal} {\bibinfo  {journal} {Nat. Commun.}\
  }\textbf {\bibinfo {volume} {14}},\ \bibinfo {pages} {3751} (\bibinfo {year}
  {2023})}\BibitemShut {NoStop}%
\bibitem [{\citenamefont {Gebhart}\ \emph {et~al.}(2023)\citenamefont
  {Gebhart}, \citenamefont {Santagati}, \citenamefont {Gentile}, \citenamefont
  {Gauger}, \citenamefont {Craig}, \citenamefont {Ares}, \citenamefont
  {Banchi}, \citenamefont {Marquardt}, \citenamefont {Pezz{\`e}},\ and\
  \citenamefont {Bonato}}]{gebhart2023learning}%
  \BibitemOpen
  \bibfield  {author} {\bibinfo {author} {\bibfnamefont {V.}~\bibnamefont
  {Gebhart}}, \bibinfo {author} {\bibfnamefont {R.}~\bibnamefont {Santagati}},
  \bibinfo {author} {\bibfnamefont {A.~A.}\ \bibnamefont {Gentile}}, \bibinfo
  {author} {\bibfnamefont {E.~M.}\ \bibnamefont {Gauger}}, \bibinfo {author}
  {\bibfnamefont {D.}~\bibnamefont {Craig}}, \bibinfo {author} {\bibfnamefont
  {N.}~\bibnamefont {Ares}}, \bibinfo {author} {\bibfnamefont {L.}~\bibnamefont
  {Banchi}}, \bibinfo {author} {\bibfnamefont {F.}~\bibnamefont {Marquardt}},
  \bibinfo {author} {\bibfnamefont {L.}~\bibnamefont {Pezz{\`e}}},\ and\
  \bibinfo {author} {\bibfnamefont {C.}~\bibnamefont {Bonato}},\ }\bibfield
  {title} {\bibinfo {title} {Learning quantum systems},\ }\href@noop {}
  {\bibfield  {journal} {\bibinfo  {journal} {Nat. Rev. Phys.}\ }\textbf
  {\bibinfo {volume} {5}},\ \bibinfo {pages} {141} (\bibinfo {year}
  {2023})}\BibitemShut {NoStop}%
\bibitem [{\citenamefont {Li}\ \emph {et~al.}(2021)\citenamefont {Li},
  \citenamefont {Ou}, \citenamefont {Lei},\ and\ \citenamefont
  {Liu}}]{li2021cavity}%
  \BibitemOpen
  \bibfield  {author} {\bibinfo {author} {\bibfnamefont {B.-B.}\ \bibnamefont
  {Li}}, \bibinfo {author} {\bibfnamefont {L.}~\bibnamefont {Ou}}, \bibinfo
  {author} {\bibfnamefont {Y.}~\bibnamefont {Lei}},\ and\ \bibinfo {author}
  {\bibfnamefont {Y.-C.}\ \bibnamefont {Liu}},\ }\bibfield  {title} {\bibinfo
  {title} {Cavity optomechanical sensing},\ }\href@noop {} {\bibfield
  {journal} {\bibinfo  {journal} {Nanophotonics}\ }\textbf {\bibinfo {volume}
  {10}},\ \bibinfo {pages} {2799} (\bibinfo {year} {2021})}\BibitemShut
  {NoStop}%
\bibitem [{\citenamefont {Liu}\ \emph {et~al.}(2021{\natexlab{b}})\citenamefont
  {Liu}, \citenamefont {Liu}, \citenamefont {Ren}, \citenamefont {Ma},
  \citenamefont {Dong}, \citenamefont {Zhou},\ and\ \citenamefont
  {Lee}}]{liu2021progress}%
  \BibitemOpen
  \bibfield  {author} {\bibinfo {author} {\bibfnamefont {X.}~\bibnamefont
  {Liu}}, \bibinfo {author} {\bibfnamefont {W.}~\bibnamefont {Liu}}, \bibinfo
  {author} {\bibfnamefont {Z.}~\bibnamefont {Ren}}, \bibinfo {author}
  {\bibfnamefont {Y.}~\bibnamefont {Ma}}, \bibinfo {author} {\bibfnamefont
  {B.}~\bibnamefont {Dong}}, \bibinfo {author} {\bibfnamefont {G.}~\bibnamefont
  {Zhou}},\ and\ \bibinfo {author} {\bibfnamefont {C.}~\bibnamefont {Lee}},\
  }\bibfield  {title} {\bibinfo {title} {Progress of optomechanical micro/nano
  sensors: a review},\ }\href@noop {} {\bibfield  {journal} {\bibinfo
  {journal} {Int. J. Optomechatronics}\ }\textbf {\bibinfo {volume} {15}},\
  \bibinfo {pages} {120} (\bibinfo {year} {2021}{\natexlab{b}})}\BibitemShut
  {NoStop}%
\bibitem [{\citenamefont {Gavartin}\ \emph {et~al.}(2012)\citenamefont
  {Gavartin}, \citenamefont {Verlot},\ and\ \citenamefont
  {Kippenberg}}]{gavartin2012hybrid}%
  \BibitemOpen
  \bibfield  {author} {\bibinfo {author} {\bibfnamefont {E.}~\bibnamefont
  {Gavartin}}, \bibinfo {author} {\bibfnamefont {P.}~\bibnamefont {Verlot}},\
  and\ \bibinfo {author} {\bibfnamefont {T.~J.}\ \bibnamefont {Kippenberg}},\
  }\bibfield  {title} {\bibinfo {title} {A hybrid on-chip optomechanical
  transducer for ultrasensitive force measurements},\ }\href@noop {} {\bibfield
   {journal} {\bibinfo  {journal} {Nature Nanotechnology}\ }\textbf {\bibinfo
  {volume} {7}},\ \bibinfo {pages} {509} (\bibinfo {year} {2012})}\BibitemShut
  {NoStop}%
\bibitem [{\citenamefont {Krause}\ \emph {et~al.}(2012)\citenamefont {Krause},
  \citenamefont {Winger}, \citenamefont {Blasius}, \citenamefont {Lin},\ and\
  \citenamefont {Painter}}]{krause2012high}%
  \BibitemOpen
  \bibfield  {author} {\bibinfo {author} {\bibfnamefont {A.~G.}\ \bibnamefont
  {Krause}}, \bibinfo {author} {\bibfnamefont {M.}~\bibnamefont {Winger}},
  \bibinfo {author} {\bibfnamefont {T.~D.}\ \bibnamefont {Blasius}}, \bibinfo
  {author} {\bibfnamefont {Q.}~\bibnamefont {Lin}},\ and\ \bibinfo {author}
  {\bibfnamefont {O.}~\bibnamefont {Painter}},\ }\bibfield  {title} {\bibinfo
  {title} {A high-resolution microchip optomechanical accelerometer},\
  }\href@noop {} {\bibfield  {journal} {\bibinfo  {journal} {Nat. Photon.}\
  }\textbf {\bibinfo {volume} {6}},\ \bibinfo {pages} {768} (\bibinfo {year}
  {2012})}\BibitemShut {NoStop}%
\bibitem [{\citenamefont {Forstner}\ \emph {et~al.}(2012)\citenamefont
  {Forstner}, \citenamefont {Prams}, \citenamefont {Knittel}, \citenamefont
  {Van~Ooijen}, \citenamefont {Swaim}, \citenamefont {Harris}, \citenamefont
  {Szorkovszky}, \citenamefont {Bowen},\ and\ \citenamefont
  {Rubinsztein-Dunlop}}]{forstner2012cavity}%
  \BibitemOpen
  \bibfield  {author} {\bibinfo {author} {\bibfnamefont {S.}~\bibnamefont
  {Forstner}}, \bibinfo {author} {\bibfnamefont {S.}~\bibnamefont {Prams}},
  \bibinfo {author} {\bibfnamefont {J.}~\bibnamefont {Knittel}}, \bibinfo
  {author} {\bibfnamefont {E.}~\bibnamefont {Van~Ooijen}}, \bibinfo {author}
  {\bibfnamefont {J.}~\bibnamefont {Swaim}}, \bibinfo {author} {\bibfnamefont
  {G.}~\bibnamefont {Harris}}, \bibinfo {author} {\bibfnamefont
  {A.}~\bibnamefont {Szorkovszky}}, \bibinfo {author} {\bibfnamefont
  {W.}~\bibnamefont {Bowen}},\ and\ \bibinfo {author} {\bibfnamefont
  {H.}~\bibnamefont {Rubinsztein-Dunlop}},\ }\bibfield  {title} {\bibinfo
  {title} {Cavity optomechanical magnetometer},\ }\href@noop {} {\bibfield
  {journal} {\bibinfo  {journal} {Phys. Rev. Lett.}\ }\textbf {\bibinfo
  {volume} {108}},\ \bibinfo {pages} {120801} (\bibinfo {year}
  {2012})}\BibitemShut {NoStop}%
\bibitem [{\citenamefont {Sikivie}(1983)}]{Sikivie:1983ip}%
  \BibitemOpen
  \bibfield  {author} {\bibinfo {author} {\bibfnamefont {P.}~\bibnamefont
  {Sikivie}},\ }\bibfield  {title} {\bibinfo {title} {{Experimental Tests of
  the Invisible Axion}},\ }\href {https://doi.org/10.1103/PhysRevLett.51.1415}
  {\bibfield  {journal} {\bibinfo  {journal} {Phys. Rev. Lett.}\ }\textbf
  {\bibinfo {volume} {51}},\ \bibinfo {pages} {1415} (\bibinfo {year}
  {1983})},\ \bibinfo {note} {[Erratum: Phys.Rev.Lett. 52, 695
  (1984)]}\BibitemShut {NoStop}%
\bibitem [{\citenamefont {Brady}\ \emph
  {et~al.}(2024{\natexlab{a}})\citenamefont {Brady}, \citenamefont {Eickbusch},
  \citenamefont {Singh}, \citenamefont {Wu},\ and\ \citenamefont
  {Zhuang}}]{brady2023advances}%
  \BibitemOpen
  \bibfield  {author} {\bibinfo {author} {\bibfnamefont {A.~J.}\ \bibnamefont
  {Brady}}, \bibinfo {author} {\bibfnamefont {A.}~\bibnamefont {Eickbusch}},
  \bibinfo {author} {\bibfnamefont {S.}~\bibnamefont {Singh}}, \bibinfo
  {author} {\bibfnamefont {J.}~\bibnamefont {Wu}},\ and\ \bibinfo {author}
  {\bibfnamefont {Q.}~\bibnamefont {Zhuang}},\ }\bibfield  {title} {\bibinfo
  {title} {Advances in bosonic quantum error correction with
  gottesman--kitaev--preskill codes: Theory, engineering and applications},\
  }\href@noop {} {\bibfield  {journal} {\bibinfo  {journal} {Prog. Quantum
  Electron.}\ ,\ \bibinfo {pages} {100496}} (\bibinfo {year}
  {2024}{\natexlab{a}})}\BibitemShut {NoStop}%
\bibitem [{\citenamefont {Hastrup}\ and\ \citenamefont
  {Andersen}(2022)}]{hastrup2022protocol}%
  \BibitemOpen
  \bibfield  {author} {\bibinfo {author} {\bibfnamefont {J.}~\bibnamefont
  {Hastrup}}\ and\ \bibinfo {author} {\bibfnamefont {U.~L.}\ \bibnamefont
  {Andersen}},\ }\bibfield  {title} {\bibinfo {title} {Protocol for generating
  optical gottesman-kitaev-preskill states with cavity qed},\ }\href@noop {}
  {\bibfield  {journal} {\bibinfo  {journal} {Phys. Rev. Lett.}\ }\textbf
  {\bibinfo {volume} {128}},\ \bibinfo {pages} {170503} (\bibinfo {year}
  {2022})}\BibitemShut {NoStop}%
\bibitem [{\citenamefont {Penrose}(1996)}]{penrose1996gravity}%
  \BibitemOpen
  \bibfield  {author} {\bibinfo {author} {\bibfnamefont {R.}~\bibnamefont
  {Penrose}},\ }\bibfield  {title} {\bibinfo {title} {On gravity's role in
  quantum state reduction},\ }\href@noop {} {\bibfield  {journal} {\bibinfo
  {journal} {General relativity and gravitation}\ }\textbf {\bibinfo {volume}
  {28}},\ \bibinfo {pages} {581} (\bibinfo {year} {1996})}\BibitemShut
  {NoStop}%
\bibitem [{\citenamefont {Zhang}\ and\ \citenamefont
  {Zhuang}(2023)}]{zhang2023energy}%
  \BibitemOpen
  \bibfield  {author} {\bibinfo {author} {\bibfnamefont {B.}~\bibnamefont
  {Zhang}}\ and\ \bibinfo {author} {\bibfnamefont {Q.}~\bibnamefont {Zhuang}},\
  }\bibfield  {title} {\bibinfo {title} {Energy-dependent barren plateau in
  bosonic variational quantum circuits},\ }\href@noop {} {\bibfield  {journal}
  {\bibinfo  {journal} {arXiv:2305.01799}\ } (\bibinfo {year}
  {2023})}\BibitemShut {NoStop}%
\bibitem [{\citenamefont {Zhuang}\ \emph {et~al.}(2019)\citenamefont {Zhuang},
  \citenamefont {Schuster}, \citenamefont {Yoshida},\ and\ \citenamefont
  {Yao}}]{zhuang2019Scrambling}%
  \BibitemOpen
  \bibfield  {author} {\bibinfo {author} {\bibfnamefont {Q.}~\bibnamefont
  {Zhuang}}, \bibinfo {author} {\bibfnamefont {T.}~\bibnamefont {Schuster}},
  \bibinfo {author} {\bibfnamefont {B.}~\bibnamefont {Yoshida}},\ and\ \bibinfo
  {author} {\bibfnamefont {N.~Y.}\ \bibnamefont {Yao}},\ }\bibfield  {title}
  {\bibinfo {title} {Scrambling and complexity in phase space},\ }\href
  {https://doi.org/10.1103/PhysRevA.99.062334} {\bibfield  {journal} {\bibinfo
  {journal} {Phys. Rev. A}\ }\textbf {\bibinfo {volume} {99}},\ \bibinfo
  {pages} {062334} (\bibinfo {year} {2019})}\BibitemShut {NoStop}%
\bibitem [{\citenamefont {G{\'o}recki}\ \emph {et~al.}(2022)\citenamefont
  {G{\'o}recki}, \citenamefont {Riccardi},\ and\ \citenamefont
  {Maccone}}]{gorecki2022quantum}%
  \BibitemOpen
  \bibfield  {author} {\bibinfo {author} {\bibfnamefont {W.}~\bibnamefont
  {G{\'o}recki}}, \bibinfo {author} {\bibfnamefont {A.}~\bibnamefont
  {Riccardi}},\ and\ \bibinfo {author} {\bibfnamefont {L.}~\bibnamefont
  {Maccone}},\ }\bibfield  {title} {\bibinfo {title} {Quantum metrology of
  noisy spreading channels},\ }\href@noop {} {\bibfield  {journal} {\bibinfo
  {journal} {Phys. Rev. Lett.}\ }\textbf {\bibinfo {volume} {129}},\ \bibinfo
  {pages} {240503} (\bibinfo {year} {2022})}\BibitemShut {NoStop}%
\bibitem [{\citenamefont {Helstrom}(1967)}]{Helstrom_1967}%
  \BibitemOpen
  \bibfield  {author} {\bibinfo {author} {\bibfnamefont {C.}~\bibnamefont
  {Helstrom}},\ }\bibfield  {title} {\bibinfo {title} {Minimum mean-squared
  error of estimates in quantum statistics},\ }\href@noop {} {\bibfield
  {journal} {\bibinfo  {journal} {Phys. Lett. A}\ }\textbf {\bibinfo {volume}
  {25}},\ \bibinfo {pages} {101} (\bibinfo {year} {1967})}\BibitemShut
  {NoStop}%
\bibitem [{\citenamefont {Helstrom}(1976)}]{Helstrom_1976}%
  \BibitemOpen
  \bibfield  {author} {\bibinfo {author} {\bibfnamefont {C.}~\bibnamefont
  {Helstrom}},\ }\href {https://books.google.com/books?id=fv9SAAAAMAAJ} {\emph
  {\bibinfo {title} {Quantum Detection and Estimation Theory}}},\ Mathematics
  in Science and Engineering : a series of monographs and textbooks\ (\bibinfo
  {publisher} {Academic Press},\ \bibinfo {year} {1976})\BibitemShut {NoStop}%
\bibitem [{\citenamefont {Gottesman}\ \emph {et~al.}(2001)\citenamefont
  {Gottesman}, \citenamefont {Kitaev},\ and\ \citenamefont
  {Preskill}}]{gkp2001}%
  \BibitemOpen
  \bibfield  {author} {\bibinfo {author} {\bibfnamefont {D.}~\bibnamefont
  {Gottesman}}, \bibinfo {author} {\bibfnamefont {A.}~\bibnamefont {Kitaev}},\
  and\ \bibinfo {author} {\bibfnamefont {J.}~\bibnamefont {Preskill}},\
  }\bibfield  {title} {\bibinfo {title} {{Encoding a qubit in an oscillator}},\
  }\href {https://doi.org/10.1103/PhysRevA.64.012310} {\bibfield  {journal}
  {\bibinfo  {journal} {Phys. Rev. A}\ }\textbf {\bibinfo {volume} {64}},\
  \bibinfo {pages} {012310} (\bibinfo {year} {2001})}\BibitemShut {NoStop}%
\bibitem [{\citenamefont {Brady}\ \emph
  {et~al.}(2024{\natexlab{b}})\citenamefont {Brady}, \citenamefont {Eickbusch},
  \citenamefont {Singh}, \citenamefont {Wu},\ and\ \citenamefont
  {Zhuang}}]{brady2023GKPrvw}%
  \BibitemOpen
  \bibfield  {author} {\bibinfo {author} {\bibfnamefont {A.~J.}\ \bibnamefont
  {Brady}}, \bibinfo {author} {\bibfnamefont {A.}~\bibnamefont {Eickbusch}},
  \bibinfo {author} {\bibfnamefont {S.}~\bibnamefont {Singh}}, \bibinfo
  {author} {\bibfnamefont {J.}~\bibnamefont {Wu}},\ and\ \bibinfo {author}
  {\bibfnamefont {Q.}~\bibnamefont {Zhuang}},\ }\bibfield  {title} {\bibinfo
  {title} {{Advances in bosonic quantum error correction with
  Gottesman–Kitaev–Preskill Codes: Theory, engineering and applications}},\
  }\href {https://doi.org/https://doi.org/10.1016/j.pquantelec.2023.100496}
  {\bibfield  {journal} {\bibinfo  {journal} {Progress in Quantum Electronics}\
  ,\ \bibinfo {pages} {100496}} (\bibinfo {year}
  {2024}{\natexlab{b}})}\BibitemShut {NoStop}%
\bibitem [{\citenamefont {Aghanim}\ \emph {et~al.}(2020)\citenamefont
  {Aghanim}, \citenamefont {Akrami}, \citenamefont {Ashdown}, \citenamefont
  {Aumont}, \citenamefont {Baccigalupi}, \citenamefont {Ballardini},
  \citenamefont {Banday}, \citenamefont {Barreiro}, \citenamefont {Bartolo},
  \citenamefont {Basak} \emph {et~al.}}]{aghanim2020planck}%
  \BibitemOpen
  \bibfield  {author} {\bibinfo {author} {\bibfnamefont {N.}~\bibnamefont
  {Aghanim}}, \bibinfo {author} {\bibfnamefont {Y.}~\bibnamefont {Akrami}},
  \bibinfo {author} {\bibfnamefont {M.}~\bibnamefont {Ashdown}}, \bibinfo
  {author} {\bibfnamefont {J.}~\bibnamefont {Aumont}}, \bibinfo {author}
  {\bibfnamefont {C.}~\bibnamefont {Baccigalupi}}, \bibinfo {author}
  {\bibfnamefont {M.}~\bibnamefont {Ballardini}}, \bibinfo {author}
  {\bibfnamefont {A.}~\bibnamefont {Banday}}, \bibinfo {author} {\bibfnamefont
  {R.}~\bibnamefont {Barreiro}}, \bibinfo {author} {\bibfnamefont
  {N.}~\bibnamefont {Bartolo}}, \bibinfo {author} {\bibfnamefont
  {S.}~\bibnamefont {Basak}}, \emph {et~al.},\ }\bibfield  {title} {\bibinfo
  {title} {Planck 2018 results-vi. cosmological parameters},\ }\href@noop {}
  {\bibfield  {journal} {\bibinfo  {journal} {Astron. Astrophys.}\ }\textbf
  {\bibinfo {volume} {641}},\ \bibinfo {pages} {A6} (\bibinfo {year}
  {2020})}\BibitemShut {NoStop}%
\bibitem [{\citenamefont {Massey}\ \emph {et~al.}(2010)\citenamefont {Massey},
  \citenamefont {Kitching},\ and\ \citenamefont {Richard}}]{massey2010dark}%
  \BibitemOpen
  \bibfield  {author} {\bibinfo {author} {\bibfnamefont {R.}~\bibnamefont
  {Massey}}, \bibinfo {author} {\bibfnamefont {T.}~\bibnamefont {Kitching}},\
  and\ \bibinfo {author} {\bibfnamefont {J.}~\bibnamefont {Richard}},\
  }\bibfield  {title} {\bibinfo {title} {The dark matter of gravitational
  lensing},\ }\href@noop {} {\bibfield  {journal} {\bibinfo  {journal} {Rep.
  Prog. Phys.}\ }\textbf {\bibinfo {volume} {73}},\ \bibinfo {pages} {086901}
  (\bibinfo {year} {2010})}\BibitemShut {NoStop}%
\bibitem [{\citenamefont {Sofue}\ and\ \citenamefont
  {Rubin}(2001)}]{sofue2001rotation}%
  \BibitemOpen
  \bibfield  {author} {\bibinfo {author} {\bibfnamefont {Y.}~\bibnamefont
  {Sofue}}\ and\ \bibinfo {author} {\bibfnamefont {V.}~\bibnamefont {Rubin}},\
  }\bibfield  {title} {\bibinfo {title} {Rotation curves of spiral galaxies},\
  }\href@noop {} {\bibfield  {journal} {\bibinfo  {journal} {Annu. Rev. Astron.
  Astr.}\ }\textbf {\bibinfo {volume} {39}},\ \bibinfo {pages} {137} (\bibinfo
  {year} {2001})}\BibitemShut {NoStop}%
\bibitem [{\citenamefont {Zheng}\ \emph {et~al.}(2016)\citenamefont {Zheng},
  \citenamefont {Silveri}, \citenamefont {Brierley}, \citenamefont {Girvin},\
  and\ \citenamefont {Lehnert}}]{girvin2016axdm}%
  \BibitemOpen
  \bibfield  {author} {\bibinfo {author} {\bibfnamefont {H.}~\bibnamefont
  {Zheng}}, \bibinfo {author} {\bibfnamefont {M.}~\bibnamefont {Silveri}},
  \bibinfo {author} {\bibfnamefont {R.~T.}\ \bibnamefont {Brierley}}, \bibinfo
  {author} {\bibfnamefont {S.~M.}\ \bibnamefont {Girvin}},\ and\ \bibinfo
  {author} {\bibfnamefont {K.~W.}\ \bibnamefont {Lehnert}},\ }\bibfield
  {title} {\bibinfo {title} {Accelerating dark-matter axion searches with
  quantum measurement technology},\ }\href@noop {} {\bibfield  {journal}
  {\bibinfo  {journal} {arXiv:1607.02529}\ } (\bibinfo {year}
  {2016})}\BibitemShut {NoStop}%
\bibitem [{\citenamefont {Malnou}\ \emph {et~al.}(2019)\citenamefont {Malnou},
  \citenamefont {Palken}, \citenamefont {Brubaker}, \citenamefont {Vale},
  \citenamefont {Hilton},\ and\ \citenamefont {Lehnert}}]{malnou2019}%
  \BibitemOpen
  \bibfield  {author} {\bibinfo {author} {\bibfnamefont {M.}~\bibnamefont
  {Malnou}}, \bibinfo {author} {\bibfnamefont {D.}~\bibnamefont {Palken}},
  \bibinfo {author} {\bibfnamefont {B.}~\bibnamefont {Brubaker}}, \bibinfo
  {author} {\bibfnamefont {L.~R.}\ \bibnamefont {Vale}}, \bibinfo {author}
  {\bibfnamefont {G.~C.}\ \bibnamefont {Hilton}},\ and\ \bibinfo {author}
  {\bibfnamefont {K.}~\bibnamefont {Lehnert}},\ }\bibfield  {title} {\bibinfo
  {title} {Squeezed vacuum used to accelerate the search for a weak classical
  signal},\ }\href {https://doi.org/10.1103/PhysRevX.9.021023} {\bibfield
  {journal} {\bibinfo  {journal} {Phys. Rev. X}\ }\textbf {\bibinfo {volume}
  {9}},\ \bibinfo {pages} {021023} (\bibinfo {year} {2019})}\BibitemShut
  {NoStop}%
\bibitem [{\citenamefont {Dixit}\ \emph {et~al.}(2021)\citenamefont {Dixit},
  \citenamefont {Chakram}, \citenamefont {He}, \citenamefont {Agrawal},
  \citenamefont {Naik}, \citenamefont {Schuster},\ and\ \citenamefont
  {Chou}}]{dixit2021}%
  \BibitemOpen
  \bibfield  {author} {\bibinfo {author} {\bibfnamefont {A.~V.}\ \bibnamefont
  {Dixit}}, \bibinfo {author} {\bibfnamefont {S.}~\bibnamefont {Chakram}},
  \bibinfo {author} {\bibfnamefont {K.}~\bibnamefont {He}}, \bibinfo {author}
  {\bibfnamefont {A.}~\bibnamefont {Agrawal}}, \bibinfo {author} {\bibfnamefont
  {R.~K.}\ \bibnamefont {Naik}}, \bibinfo {author} {\bibfnamefont {D.~I.}\
  \bibnamefont {Schuster}},\ and\ \bibinfo {author} {\bibfnamefont
  {A.}~\bibnamefont {Chou}},\ }\bibfield  {title} {\bibinfo {title} {Searching
  for dark matter with a superconducting qubit},\ }\href
  {https://doi.org/10.1103/PhysRevLett.126.141302} {\bibfield  {journal}
  {\bibinfo  {journal} {Phys. Rev. Lett.}\ }\textbf {\bibinfo {volume} {126}},\
  \bibinfo {pages} {141302} (\bibinfo {year} {2021})}\BibitemShut {NoStop}%
\bibitem [{\citenamefont {Backes}\ \emph
  {et~al.}(2021{\natexlab{b}})\citenamefont {Backes}, \citenamefont {Palken},
  \citenamefont {Al~Kenany}, \citenamefont {Brubaker}, \citenamefont {Cahn},
  \citenamefont {Droster}, \citenamefont {Hilton}, \citenamefont {Ghosh},
  \citenamefont {Jackson}, \citenamefont {Lamoreaux} \emph
  {et~al.}}]{backes2021}%
  \BibitemOpen
  \bibfield  {author} {\bibinfo {author} {\bibfnamefont {K.}~\bibnamefont
  {Backes}}, \bibinfo {author} {\bibfnamefont {D.}~\bibnamefont {Palken}},
  \bibinfo {author} {\bibfnamefont {S.}~\bibnamefont {Al~Kenany}}, \bibinfo
  {author} {\bibfnamefont {B.}~\bibnamefont {Brubaker}}, \bibinfo {author}
  {\bibfnamefont {S.}~\bibnamefont {Cahn}}, \bibinfo {author} {\bibfnamefont
  {A.}~\bibnamefont {Droster}}, \bibinfo {author} {\bibfnamefont {G.~C.}\
  \bibnamefont {Hilton}}, \bibinfo {author} {\bibfnamefont {S.}~\bibnamefont
  {Ghosh}}, \bibinfo {author} {\bibfnamefont {H.}~\bibnamefont {Jackson}},
  \bibinfo {author} {\bibfnamefont {S.}~\bibnamefont {Lamoreaux}}, \emph
  {et~al.},\ }\bibfield  {title} {\bibinfo {title} {A quantum enhanced search
  for dark matter axions},\ }\href {https://doi.org/10.1038/s41586-021-03226-7}
  {\bibfield  {journal} {\bibinfo  {journal} {Nature}\ }\textbf {\bibinfo
  {volume} {590}},\ \bibinfo {pages} {238} (\bibinfo {year}
  {2021}{\natexlab{b}})}\BibitemShut {NoStop}%
\bibitem [{\citenamefont {Ghelfi}\ \emph {et~al.}(2014)\citenamefont {Ghelfi},
  \citenamefont {Laghezza}, \citenamefont {Scotti}, \citenamefont {Serafino},
  \citenamefont {Capria}, \citenamefont {Pinna}, \citenamefont {Onori},
  \citenamefont {Porzi}, \citenamefont {Scaffardi}, \citenamefont {Malacarne}
  \emph {et~al.}}]{ghelfi2014fully}%
  \BibitemOpen
  \bibfield  {author} {\bibinfo {author} {\bibfnamefont {P.}~\bibnamefont
  {Ghelfi}}, \bibinfo {author} {\bibfnamefont {F.}~\bibnamefont {Laghezza}},
  \bibinfo {author} {\bibfnamefont {F.}~\bibnamefont {Scotti}}, \bibinfo
  {author} {\bibfnamefont {G.}~\bibnamefont {Serafino}}, \bibinfo {author}
  {\bibfnamefont {A.}~\bibnamefont {Capria}}, \bibinfo {author} {\bibfnamefont
  {S.}~\bibnamefont {Pinna}}, \bibinfo {author} {\bibfnamefont
  {D.}~\bibnamefont {Onori}}, \bibinfo {author} {\bibfnamefont
  {C.}~\bibnamefont {Porzi}}, \bibinfo {author} {\bibfnamefont
  {M.}~\bibnamefont {Scaffardi}}, \bibinfo {author} {\bibfnamefont
  {A.}~\bibnamefont {Malacarne}}, \emph {et~al.},\ }\bibfield  {title}
  {\bibinfo {title} {A fully photonics-based coherent radar system},\
  }\href@noop {} {\bibfield  {journal} {\bibinfo  {journal} {Nature}\ }\textbf
  {\bibinfo {volume} {507}},\ \bibinfo {pages} {341} (\bibinfo {year}
  {2014})}\BibitemShut {NoStop}%
\bibitem [{\citenamefont {Skovgaard}(1951)}]{skovgaard1951greatest}%
  \BibitemOpen
  \bibfield  {author} {\bibinfo {author} {\bibfnamefont {H.}~\bibnamefont
  {Skovgaard}},\ }\bibfield  {title} {\bibinfo {title} {On the greatest and the
  least zero of laguerre polynomials},\ }\href@noop {} {\bibfield  {journal}
  {\bibinfo  {journal} {Matematisk tidsskrift. B}\ ,\ \bibinfo {pages} {59}}
  (\bibinfo {year} {1951})}\BibitemShut {NoStop}%
\bibitem [{\citenamefont {Becerra}\ \emph {et~al.}(2013)\citenamefont
  {Becerra}, \citenamefont {Fan}, \citenamefont {Baumgartner}, \citenamefont
  {Goldhar}, \citenamefont {Kosloski},\ and\ \citenamefont
  {Migdall}}]{becerra2013experimental}%
  \BibitemOpen
  \bibfield  {author} {\bibinfo {author} {\bibfnamefont {F.}~\bibnamefont
  {Becerra}}, \bibinfo {author} {\bibfnamefont {J.}~\bibnamefont {Fan}},
  \bibinfo {author} {\bibfnamefont {G.}~\bibnamefont {Baumgartner}}, \bibinfo
  {author} {\bibfnamefont {J.}~\bibnamefont {Goldhar}}, \bibinfo {author}
  {\bibfnamefont {J.}~\bibnamefont {Kosloski}},\ and\ \bibinfo {author}
  {\bibfnamefont {A.}~\bibnamefont {Migdall}},\ }\bibfield  {title} {\bibinfo
  {title} {Experimental demonstration of a receiver beating the standard
  quantum limit for multiple nonorthogonal state discrimination},\ }\href@noop
  {} {\bibfield  {journal} {\bibinfo  {journal} {Nat. Photon.}\ }\textbf
  {\bibinfo {volume} {7}},\ \bibinfo {pages} {147} (\bibinfo {year}
  {2013})}\BibitemShut {NoStop}%
\bibitem [{\citenamefont {Cui}\ \emph {et~al.}(2022)\citenamefont {Cui},
  \citenamefont {Horrocks}, \citenamefont {Hao}, \citenamefont {Guha},
  \citenamefont {Peyghambarian}, \citenamefont {Zhuang},\ and\ \citenamefont
  {Zhang}}]{cui2022quantum}%
  \BibitemOpen
  \bibfield  {author} {\bibinfo {author} {\bibfnamefont {C.}~\bibnamefont
  {Cui}}, \bibinfo {author} {\bibfnamefont {W.}~\bibnamefont {Horrocks}},
  \bibinfo {author} {\bibfnamefont {S.}~\bibnamefont {Hao}}, \bibinfo {author}
  {\bibfnamefont {S.}~\bibnamefont {Guha}}, \bibinfo {author} {\bibfnamefont
  {N.}~\bibnamefont {Peyghambarian}}, \bibinfo {author} {\bibfnamefont
  {Q.}~\bibnamefont {Zhuang}},\ and\ \bibinfo {author} {\bibfnamefont
  {Z.}~\bibnamefont {Zhang}},\ }\bibfield  {title} {\bibinfo {title} {Quantum
  receiver enhanced by adaptive learning},\ }\href@noop {} {\bibfield
  {journal} {\bibinfo  {journal} {Light Sci. Appl.}\ }\textbf {\bibinfo
  {volume} {11}},\ \bibinfo {pages} {344} (\bibinfo {year} {2022})}\BibitemShut
  {NoStop}%
\bibitem [{\citenamefont {Kaubruegger}\ \emph {et~al.}(2019)\citenamefont
  {Kaubruegger}, \citenamefont {Silvi}, \citenamefont {Kokail}, \citenamefont
  {van Bijnen}, \citenamefont {Rey}, \citenamefont {Ye}, \citenamefont
  {Kaufman},\ and\ \citenamefont {Zoller}}]{kaubruegger2019variational}%
  \BibitemOpen
  \bibfield  {author} {\bibinfo {author} {\bibfnamefont {R.}~\bibnamefont
  {Kaubruegger}}, \bibinfo {author} {\bibfnamefont {P.}~\bibnamefont {Silvi}},
  \bibinfo {author} {\bibfnamefont {C.}~\bibnamefont {Kokail}}, \bibinfo
  {author} {\bibfnamefont {R.}~\bibnamefont {van Bijnen}}, \bibinfo {author}
  {\bibfnamefont {A.~M.}\ \bibnamefont {Rey}}, \bibinfo {author} {\bibfnamefont
  {J.}~\bibnamefont {Ye}}, \bibinfo {author} {\bibfnamefont {A.~M.}\
  \bibnamefont {Kaufman}},\ and\ \bibinfo {author} {\bibfnamefont
  {P.}~\bibnamefont {Zoller}},\ }\bibfield  {title} {\bibinfo {title}
  {Variational spin-squeezing algorithms on programmable quantum sensors},\
  }\href {https://doi.org/10.1103/PhysRevLett.123.260505} {\bibfield  {journal}
  {\bibinfo  {journal} {Phys. Rev. Lett.}\ }\textbf {\bibinfo {volume} {123}},\
  \bibinfo {pages} {260505} (\bibinfo {year} {2019})}\BibitemShut {NoStop}%
\bibitem [{\citenamefont {Meyer}\ \emph {et~al.}(2021)\citenamefont {Meyer},
  \citenamefont {Borregaard},\ and\ \citenamefont
  {Eisert}}]{meyer2021variational}%
  \BibitemOpen
  \bibfield  {author} {\bibinfo {author} {\bibfnamefont {J.~J.}\ \bibnamefont
  {Meyer}}, \bibinfo {author} {\bibfnamefont {J.}~\bibnamefont {Borregaard}},\
  and\ \bibinfo {author} {\bibfnamefont {J.}~\bibnamefont {Eisert}},\
  }\bibfield  {title} {\bibinfo {title} {A variational toolbox for quantum
  multi-parameter estimation},\ }\href@noop {} {\bibfield  {journal} {\bibinfo
  {journal} {npj Quantum Inf.}\ }\textbf {\bibinfo {volume} {7}},\ \bibinfo
  {pages} {89} (\bibinfo {year} {2021})}\BibitemShut {NoStop}%
\bibitem [{\citenamefont {Kaubruegger}\ \emph {et~al.}(2023)\citenamefont
  {Kaubruegger}, \citenamefont {Shankar}, \citenamefont {Vasilyev},\ and\
  \citenamefont {Zoller}}]{kaubruegger2023optimal}%
  \BibitemOpen
  \bibfield  {author} {\bibinfo {author} {\bibfnamefont {R.}~\bibnamefont
  {Kaubruegger}}, \bibinfo {author} {\bibfnamefont {A.}~\bibnamefont
  {Shankar}}, \bibinfo {author} {\bibfnamefont {D.~V.}\ \bibnamefont
  {Vasilyev}},\ and\ \bibinfo {author} {\bibfnamefont {P.}~\bibnamefont
  {Zoller}},\ }\bibfield  {title} {\bibinfo {title} {Optimal and variational
  multiparameter quantum metrology and vector-field sensing},\ }\href
  {https://doi.org/10.1103/PRXQuantum.4.020333} {\bibfield  {journal} {\bibinfo
   {journal} {PRX Quantum}\ }\textbf {\bibinfo {volume} {4}},\ \bibinfo {pages}
  {020333} (\bibinfo {year} {2023})}\BibitemShut {NoStop}%
\bibitem [{\citenamefont {Sinanan-Singh}\ \emph {et~al.}(2023)\citenamefont
  {Sinanan-Singh}, \citenamefont {Mintzer}, \citenamefont {Chuang},\ and\
  \citenamefont {Liu}}]{sinanan2023single}%
  \BibitemOpen
  \bibfield  {author} {\bibinfo {author} {\bibfnamefont {J.}~\bibnamefont
  {Sinanan-Singh}}, \bibinfo {author} {\bibfnamefont {G.~L.}\ \bibnamefont
  {Mintzer}}, \bibinfo {author} {\bibfnamefont {I.~L.}\ \bibnamefont
  {Chuang}},\ and\ \bibinfo {author} {\bibfnamefont {Y.}~\bibnamefont {Liu}},\
  }\bibfield  {title} {\bibinfo {title} {Single-shot quantum signal processing
  interferometry},\ }\href@noop {} {\bibfield  {journal} {\bibinfo  {journal}
  {arXiv:2311.13703}\ } (\bibinfo {year} {2023})}\BibitemShut {NoStop}%
\bibitem [{\citenamefont {Oh}\ \emph {et~al.}(2024)\citenamefont {Oh},
  \citenamefont {Chen}, \citenamefont {Wong}, \citenamefont {Zhou},
  \citenamefont {Huang}, \citenamefont {Nielsen}, \citenamefont {Liu},
  \citenamefont {Neergaard-Nielsen}, \citenamefont {Andersen}, \citenamefont
  {Jiang} \emph {et~al.}}]{oh2024entanglement}%
  \BibitemOpen
  \bibfield  {author} {\bibinfo {author} {\bibfnamefont {C.}~\bibnamefont
  {Oh}}, \bibinfo {author} {\bibfnamefont {S.}~\bibnamefont {Chen}}, \bibinfo
  {author} {\bibfnamefont {Y.}~\bibnamefont {Wong}}, \bibinfo {author}
  {\bibfnamefont {S.}~\bibnamefont {Zhou}}, \bibinfo {author} {\bibfnamefont
  {H.-Y.}\ \bibnamefont {Huang}}, \bibinfo {author} {\bibfnamefont {J.~A.}\
  \bibnamefont {Nielsen}}, \bibinfo {author} {\bibfnamefont {Z.-H.}\
  \bibnamefont {Liu}}, \bibinfo {author} {\bibfnamefont {J.~S.}\ \bibnamefont
  {Neergaard-Nielsen}}, \bibinfo {author} {\bibfnamefont {U.~L.}\ \bibnamefont
  {Andersen}}, \bibinfo {author} {\bibfnamefont {L.}~\bibnamefont {Jiang}},
  \emph {et~al.},\ }\bibfield  {title} {\bibinfo {title} {Entanglement-enabled
  advantage for learning a bosonic random displacement channel},\ }\href@noop
  {} {\bibfield  {journal} {\bibinfo  {journal} {arXiv:2402.18809}\ } (\bibinfo
  {year} {2024})}\BibitemShut {NoStop}%
\end{thebibliography}
%

\end{document}